\documentclass[a4paper]{article} 

\usepackage{amsthm}
\usepackage{amsmath}
\usepackage{amssymb}
\usepackage{multicol}

\newtheorem{thm}{Theorem}
\newtheorem{cor}{Corollary}
\newtheorem{lem}{Lemma}
\newtheorem{defn}{Definition}
\newtheorem{exmp}{Example}
\newtheorem{res}{Restriction}

\title{CCS-Based Dynamic Logics for Communicating Concurrent Programs\footnote{This work was supported by the Brazilian research agencies CNPq and CAPES. A preliminary version of this work was published in the proceedings of WoLLIC 2008 \cite{Wollic}.}}
\author{Mario R. F. Benevides\footnote{Systems and Computer Engineering Program and Computer Science Department, Federal University of Rio de Janeiro, Brazil} \and L. Menasch\'e Schechter\footnote{Systems and Computer Engineering Program, Federal University of Rio de Janeiro, Brazil}}
\date{\{mario,luis\}@cos.ufrj.br}

\begin{document}

\maketitle

\begin{abstract}
This work presents three increasingly expressive Dynamic Logics in which the programs are CCS processes (sCCS-PDL, CCS-PDL and XCCS-PDL). Their goal is to reason about properties of concurrent programs and systems described using CCS. In order to accomplish that, CCS's operators and constructions are added to a basic modal logic in order to create dynamic logics that are suitable for the description and verification of properties of communicating, concurrent and non-deterministic programs and systems, in a similar way as PDL is used for the sequential case. We provide complete axiomatizations for the three logics. Unlike Peleg's Concurrent PDL with Channels, our logics have a simple Kripke semantics, complete axiomatizations and the finite model property.
\end{abstract}


\section{Introduction}

Propositional Dynamic Logic (PDL) \cite{fl} plays an important role in formal specification and reasoning about sequential programs and systems. PDL is a multi-modal logic with one modality $\langle \pi \rangle$ for each program $\pi$. The logic has a set of basic programs and a set of operators (sequential composition, iteration and nondeterministic choice) that are used to inductively build the set of non-basic programs. PDL has been used to describe and verify properties and behaviour of sequential programs and systems. Correctness, termination, fairness, liveness and equivalence of programs are among the properties that one usually wants to verify. A Kripke semantics can be provided, with a frame ${\mathcal F} = (W , R_{\pi})$, where $W$ is a non-empty set of possible program states and, for each program $\pi$, $R_{\pi}$ is a binary relation on $W$ such that $(s,t) \in R_{\pi}$ if and only if there is a computation of $\pi$ starting in $s$ and terminating in $t$.

The Calculus for Communicating Systems (CCS) is a well known process algebra, proposed by Robin Milner \cite{M89}, for the specification of communicating concurrent systems. It models the concurrency and interaction between processes through individual acts of communication. A pair of processes can communicate through a common channel and each act of communication consists simply of a signal being sent at one end of the channel and immediately being received at the other. A CCS specification is a description (in the form of algebraic equations) of the behaviour expected from a system, based on the communication events that may occur. As in PDL, CCS has a set of operators (action prefix, parallel composition, nondeterministic choice and restriction on acts of communication) that are used to inductively build process specifications from a set of basic actions. Iteration can also be described through the use of recursive equations.

This work presents three increasingly expressive Dynamic Logics in which the programs are CCS processes (sCCS-PDL, CCS-PDL and XCCS-PDL). Their goal is to reason about properties of concurrent programs and systems described using CCS.

There are, in the literature, some logics that make use of CCS or other process algebras. However, they use these process algebras as a language for the description of frames and models, while using standard modal logics for the description of properties (see, for example, \cite{M89} and \cite{MilParWalLogic}). The logics that we develop in the present work use CCS in a distinct way. Its operators and constructions are \emph{added} to a basic modal logic in order to create dynamic logics that are suitable for the description and verification of properties of communicating, concurrent and non-deterministic programs and systems, in a similar way as PDL is used for the sequential case.

Thus, it should be emphasized that the contribution of this work is on the field of dynamic logics and not on the field of process algebras. From process algebras, we just borrow a set of operators that are suitable for the description of communication and concurrency. We use these operators because they have a well-established theory behind them and we can use many of its concepts and results to help us build our logics.

Our paper falls in the broad category of works that attempt to generalize PDL and build dynamic logics that deal with classes of non-regular programs. As examples of other works in this area, we can mention \cite{CFPDL1}, \cite{CFPDL2} and \cite{CFPDL3}, that develop decidable dynamic logics for fragments of the class of context-free programs and \cite{DPE87}, \cite{DPEL87} and \cite{Vera}, that develop dynamic logics for classes of programs with some sort of concurrency. Our logics have a close relation to two logics in this last group: Concurrent PDL with Channels \cite{DPEL87} and the logic developed in \cite{Vera}. Both of these logics are expressive enough to represent interesting properties of communicating concurrent systems. However, neither of them has a simple Kripke semantics. The first has a semantics based on \emph{super-states} and \emph{super-processes} and its satisfiability problem can be proved undecidable (in fact, it is $\Pi_{1}^{1}$-hard). Also, it does not have a complete axiomatization \cite {DPEL87}. The second makes a semantic distinction between \emph{final} and \emph{non-final} states, which makes its semantics and its axiomatization rather complex. On the other hand, due to the use of the CCS mechanisms of communication and concurrency, our logics have a simple Kripke semantics, simple and complete axiomatizations and the finite model property.

We choose to base our logics in the mechanisms of communication and concurrency of CCS, instead of some other process algebra, for two reasons. First, CCS is built with the philosophy that only those operators that are essential to the description of the basic behaviours of communication and concurrency should be included as primitives in the language, while the operators and behaviours of greater complexity should be derived from the basic ones. Using a small language like CCS, where only the more basic constructions are present, we can study in details what are the problems that may arise when we try to use its operators to build a dynamic logic and what operators and constructions we need to add or remove to correct these problems. Second, the development of CCS-based dynamic logics can be used as a natural stepping stone to the development of dynamic logics based on the $\pi$-Calculus \cite{M99}, a very powerful process algebra that is able to describe not only non-determinism and concurrency, but also \emph{mobility} of processes. The $\pi$-Calculus can also be used to encode some powerful programming paradigms, as object-oriented programming and functional programming ($\lambda$-Calculus) \cite{M99}.

The rest of this paper is organized as follows. In section \ref{ccs}, we introduce the necessary background concepts: Propositional Dynamic Logic and the Calculus for Communicating Systems. Our first logic (sCCS-PDL), together with a complete axiomatic system, is presented in section \ref{sec:sccs}. In this logic, we do not use constants or restriction in the CCS processes. In section \ref{pdl-ccs}, we present our second logic (CCS-PDL), in which we allow the presence of constants in the CCS processes. We also give an axiomatization for this second logic and prove its completeness using a Fischer-Ladner construction. The third logic (XCCS-PDL), together with a complete axiomatization for it, is presented in section \ref{sec:xccs}. In this logic, we extend CCS with some extra operators, which allows us to solve some issues that appear in the previous logics. Finally, in section \ref{DC}, we state our final remarks. 

In the preliminary version of this work (\cite{Wollic}), the contents of section \ref{sec:xccs} are completely absent and the concepts and proofs in section \ref{pdl-ccs} are presented with far less details. Besides that, most of the motivations, discussions and detailed explanations that we present in this paper, trying to show what guided our choices in the construction of these logics, are also absent from \cite{Wollic}.

\section{Background}\label{ccs}

This section presents two important subjects. First, we make a brief review of the syntax and semantics of PDL. Second, we present the process algebra CCS together with some useful concepts, properties and results from its theory. We do not assume a familiarity with CCS, since process algebras are by no means a universally studied topic among (modal) logicians. We introduce here all that is necessary for our presentation in the next sections, trying to make this work as self-contained as possible.

\subsection{Propositional Dynamic Logic}\label{PDL}

In this section, we present the syntax and semantics of PDL.

\begin{defn}\label{def-langpdl} 
The PDL language consists of a set $\Phi$ of countably many proposition symbols, a set $\Pi$ of countably many basic programs, the boolean connectives $\neg$ and $\land$, the program constructors $;$, $\cup$ and $\phantom{}^*$ and a modality $\langle \pi \rangle$ for every program $\pi$. The formulas are defined as follows:
\[
\varphi ::= p \mid \top \mid \neg \varphi \mid \varphi_1 \wedge \varphi_2 \mid \langle \pi \rangle \varphi, \mbox{ with } \, \pi ::= a \mid \pi_1;\pi_2 \mid \pi_1 \cup \pi_2 \mid \pi^{*},
\]
where $p \in \Phi$ and $a \in \Pi$.
\end{defn}

In all the logics that appear in this paper, we use the standard abbreviations $\bot \equiv \neg \top$, $\varphi \lor \phi \equiv \neg ( \neg \varphi \land \neg \phi)$, $\varphi \rightarrow \phi \equiv \neg( \varphi \land \neg \phi)$ and $[\pi] \varphi \equiv \neg \langle \pi \rangle \neg \varphi$.

\begin{defn}\label{def-framepdl} 
A \emph{frame} for PDL is a tuple $\mathcal{F}= (W, R_{a}, R_{\pi})$ where
\begin{itemize}
\item $W$ is a non-empty set of states;
\item $R_{a}$ is a binary relation for each basic program $a$;
\item $R_{\pi}$ is a binary relation for each non-basic program $\pi$, inductively built using the rules $R_{\pi_1 ; \pi_2} = R_{\pi_1} \circ R_{\pi_2}$, $R_{\pi_1 \cup \pi_2} = R_{\pi_1} \cup R_{\pi_2}$ and $R_{\pi^{*}} = R_{\pi}^{*}$, where $R_{\pi}^{*}$ denotes the reflexive transitive closure of $R_{\pi}$.
\end{itemize}
\end{defn}

\begin{defn}\label{def-modelpdl} 
A \emph{model} for PDL is a pair $\mathcal{M}= ({\cal F}, {\bf V})$, where ${\cal F}$ is a PDL frame and ${\bf V}$  is a valuation function ${\bf V} : \Phi \mapsto 2^W$.
\end{defn}

The semantical notion of satisfaction for PDL is defined as follows:

\begin{defn}\label{def-satpdl} 
Let $\mathcal{M}= ({\cal F}, {\bf V})$ be a model. The notion of \emph{satisfaction} of a formula $\varphi$ in a model $\mathcal{M}$ at a state $w$, notation $\mathcal{M},w \Vdash \varphi$, can be inductively defined as follows:
\begin{itemize}
\item $\mathcal{M},w \Vdash p$ iff $w \in {\bf V}(p)$;
\item $\mathcal{M},w \Vdash \top$ always;
\item $\mathcal{M},w \Vdash \neg \varphi$ iff $\mathcal{M},w \not\Vdash \varphi$;
\item $\mathcal{M},w \Vdash \varphi_{1} \wedge \varphi_{2}$ iff $\mathcal{M},w \Vdash \varphi_{1}$ and $\mathcal{M},w \Vdash \varphi_{2}$;
\item $\mathcal{M},w \Vdash \langle \pi \rangle \varphi$ iff there is $w' \in W$ such that $w R_{\pi} w'$ and $\mathcal{M},w' \Vdash \varphi$.
\end{itemize}
\end{defn}

\subsection{Calculus for Communicating Systems}\label{CP}

The Calculus for Communicating Systems (CCS) is a well known process algebra, proposed by Robin Milner \cite{M89}, for the specification of communicating concurrent systems. It models the concurrency and interaction between processes through individual acts of communication. A CCS specification is a description (in the form of algebraic equations) of the behaviour expected from a system, based on the communication events that may occur. For a broad introduction to CCS, \cite{M89} can be consulted.

In CCS, a pair of processes can communicate through a common channel and each act of communication consists simply of a signal being sent at one end of the channel and immediately being received at the other.

Let ${\cal N} = \{a,b,c,\ldots\}$  be a set of names. Each channel in a CCS specification is labelled by a name. The labels of the channels are also used to describe the communication actions (sending and receiving signals) performed by the processes, as is shown below. Besides these communication actions, CCS has only one other action: the silent action, denoted by $\tau$, used to represent any internal action performed by any of the processes that does not involve an act of communication (e.g.: a memory update).

There are two possible semantics for the $\tau$ action in CCS: it can be regarded as being observable, in the same way as the communication actions, or it can be regarded as being invisible. We adopt the first one, since it is more generic. In our logical formalism, we are able to represent the second semantics as a particular case of the first.

\begin{defn}\label{def:ccslang}
In our presentation of CCS, process specifications can be built using the following operations:
\[
P ::=  \alpha \mid \alpha.P \mid \alpha.A \mid P_1 + P_2 \mid P_1 | P_2 \mid P \backslash L,
\]
with
\[
\alpha ::= a \mid \overline{a} \mid \tau,
\]
where $a \in {\cal N}$, $L \subseteq {\cal N}$ and every constant $A$ has a unique \emph{defining equation} $A \stackrel{def}{=} P_A$, where $P_A$ is a process specification. In this work, every time that a process is linked to a constant $A$ through a defining equation, it will be denoted by $P_A$. 
\end{defn}

Originally, CCS also defines a null process, denoted by ${\bf 0}$. It represents the process that is unable to perform any actions. However, because of its somewhat loose definition, which fails to differentiate between a deadlock and a successful termination (unlike other process algebras, as ACP \cite{Fokkink} for instance, in which the deadlocked process and the terminated process are different), its use would bring a serious inconvenience to the semantics of our first two logics: the semantics would not be fully compositional. This is shown in details in the next section. Because of that, we drop this null process until our third logic, when we extend CCS with new operators and partially redefine its semantics, obtaining a null process with a much better algebraic behaviour. To completely drop the null process, we must also drop the restriction operator, as it may be used to define such a process (e.g. $a \backslash \{a\}$). Hence, the restriction operator will also only be present in our third logic.

The \emph{prefix} operator (.) denotes that the process will first perform the action $\alpha$ and then behave as $P$ or $A$. The \emph{summation} (or \emph{nondeterministic choice}) operator (+) denotes that the process will make a nondeterministic choice to behave as either $P_1$ or $P_2$. The \emph{parallel composition} operator ($\mid$) denotes that the processes $P_1$ and $P_2$ may proceed independently or may communicate through a common channel. Finally, the \emph{restriction} operator ($\backslash$) denotes that the channels in $L$ are only accessible inside $P$. Iteration in CCS is modeled through recursive defining equations, i.e., equations $A \stackrel{def}{=} P_A$ where $A$ occurs in $P_A$.

The action $a$, called \emph{input action}, denotes that the process receives a signal through the channel labelled by $a$. The action $\overline{a}$, called \emph{output action}, denotes that the process sends a signal through the channel labelled by $a$. Finally, $\tau$ denotes the silent action.

We write $P \stackrel{\alpha}{\rightarrow} P'$ to express that the process $P$ can perform the action $\alpha$ and after that behave as $P'$. We write $P \stackrel{\alpha}{\rightarrow} \surd$ to express that the process $P$ successfully finishes after performing the action $\alpha$ (a notation borrowed from ACP). A process only finishes when there is not any possible action left for it to perform. For example, $\beta \stackrel{\beta}{\rightarrow} \surd$. When a process finishes inside a parallel composition, we write $P$ instead of $P| \surd$. We also write $\surd$ instead of $\surd \backslash L$ and $\surd | \surd$. We define the set $\overline{L}$ as $\overline{L} = \{ \overline{a} : a \in L \}$. In table \ref{tab:semccs}, we present the semantics for the operators based on this notation. In this table, $P$, $Q$ and $P_A$ are process specifications, while $P'$ and $Q'$ are process specifications or $\surd$.

\begin{table}
\centering
\caption{Transition Relations of CCS}
\begin{tabular}{|c|c|c|c|c|}
\hline
$\alpha \stackrel{\alpha}{\rightarrow} \surd$ & 
$\alpha.P \stackrel{\alpha}{\rightarrow} P$ &
$\frac{A \stackrel{def}{=} P_A}{\alpha.A \stackrel{\alpha}{\rightarrow} P_A}$ &
$\frac{P \stackrel{\alpha}{\rightarrow} P'} {P + Q \stackrel{\alpha}{\rightarrow} P'}$ &  
$\frac{Q \stackrel{\beta}{\rightarrow} Q'} {P + Q \stackrel{\beta}{\rightarrow} Q'}$ \\
\hline
$\frac{P \stackrel{\alpha}{\rightarrow} P'}{P | Q \stackrel{\alpha}{\rightarrow} P' | Q}$ &
$\frac{Q \stackrel{\beta}{\rightarrow} Q'}{P | Q \stackrel{\beta}{\rightarrow} P | Q'}$ &
$\frac{P \stackrel{\lambda}{\rightarrow} P', Q \stackrel{\overline{\lambda}}{\rightarrow} Q'}{P | Q \stackrel{\tau}{\rightarrow} P' | Q'}$ &
$\frac{P \stackrel{\alpha}{\rightarrow} P',\alpha \not\in L \cup \overline{L}}{P \backslash L 
\stackrel{\alpha}{\rightarrow} P' \backslash L}$ & \\
\hline
\end{tabular}
\label{tab:semccs}
\end{table}

In order to motivate the use of CCS, we present a simple example of the use of the language below. Here, we are still using CCS outside of the logical formalisms that are presented in the next sections.

\begin{exmp}[\cite{M89,S94}]\label{V}
Consider a vending machine where one can put coins of one or two euro and buy a little or a big chocolate bar. After inserting the coins, one must press the little button for a little chocolate or the big button for a big chocolate. The machine is also programmed to shutdown on its own following some internal protocol (represented by a $\tau$ action). A CCS term describing the behaviour of this machine is the following:
\[
V = 1e.little.\overline{collect}.A + 1e.1e.big.\overline{collect}.A + 2e.big.\overline{collect}.A
\]
\[
A \stackrel{def}{=} 1e.little.\overline{collect}.A + 1e.1e.big.\overline{collect}.A + 2e.big.\overline{collect}.A + \tau
\]
Let us now suppose that Chuck wants to use this vending machine. We could describe Chuck as
\[
C = \overline{1e}.\overline{little}.collect + \overline{1e}.\overline{1e}.\overline{big}.collect + \overline{2e}.\overline{big}.collect.
\]
Notice that Chuck does not have an iterative behaviour. Once he collects the chocolate, he is done. Now, if we want to model the process of Chuck buying a chocolate from the vending machine, we could write $(V | C) \backslash L$, where $L = \{ 1e, 2e, little, big, collect \}$.
\end{exmp}

\begin{defn}
Let ${\cal P}$ be the set of all possible process specifications. A set $Z \subseteq {\cal P} \times {\cal P}$ is a \emph{strong bisimulation} if $(P,Q) \in Z$ implies the following:
\begin{itemize}
 \item If $P \stackrel{\alpha}{\rightarrow} P'$ and $P' \in {\cal P}$, then there is $Q' \in {\cal P}$ such that $Q \stackrel{\alpha}{\rightarrow} Q'$ and $(P',Q') \in Z$;
 \item If $Q \stackrel{\alpha}{\rightarrow} Q'$ and $Q' \in {\cal P}$, then there is $P' \in {\cal P}$ such that $P \stackrel{\alpha}{\rightarrow} P'$ and $(P',Q') \in Z$;
 \item $P \stackrel{\alpha}{\rightarrow} \surd$ if and only if $Q \stackrel{\alpha}{\rightarrow} \surd$.
\end{itemize}
\end{defn}

\begin{defn}\label{def:bisim}
Two process specifications $P$ and $Q$ are \emph{strongly bisimilar} (or simply \emph{bisimilar}), denoted by $P \sim Q$, if there is a strong bisimulation $Z$ such that $(P,Q) \in Z$.
\end{defn}

Now, we introduce the Expansion Law, which is very important in the definition of the semantics of our logics in the next sections and in their axiomatizations. We present a particular case of the Expansion Law, which is suited to our needs. The most general case of the Expansion Law is presented in \cite{M89}.

\begin{defn}\label{def:unres}
We say that a process is \emph{unrestricted} if it has no occurrences of the $\backslash$ operator.
\end{defn}

\begin{thm}[Expansion Law (EL)]\label{teo:EL}
Let $P = P_1 \mid P_2$, where $P$ is unrestricted. Then
\[
P \sim \sum_{P_1 \stackrel{\alpha}{\rightarrow} P_1'} \alpha.(P_1' \mid P_2) + \sum_{P_2 \stackrel{\beta}{\rightarrow} P_2'} \beta.(P_1 \mid P_2') + \sum_{R \in A_{\tau}} \tau . R,
\]
where $A_{\tau} = \{ (P_1' \mid P_2') : P_1 \stackrel{a}{\rightarrow} P_1' \mbox{ and } P_2 \stackrel{\overline{a}}{\rightarrow} P_2', \mbox{ for some } a \in {\cal N} \} \cup \{ (P_1' \mid P_2') : P_1 \stackrel{\overline{a}}{\rightarrow} P_1' \mbox{ and } P_2 \stackrel{a}{\rightarrow} P_2', \mbox{ for some } a \in {\cal N} \}$. We denote the right side of this bisimilarity by $Exp(P)$.
\end{thm} 

\subsection{Action Sequences and Possible Runs}

In this section, we introduce the key concept of \emph{finite possible runs} of a process. This concept plays a central role in the semantics of our logics.

\begin{defn}
We use the notation $\overrightarrow{\alpha}$ to denote a potentially infinite sequence of actions $\alpha_1 . \alpha_2 .$ $.\cdots . \alpha_n (. \cdots)$ (the empty sequence is denoted by $\overrightarrow{\varepsilon}$). The empty sequence follows the rule $\overrightarrow{\alpha} . \overrightarrow{\varepsilon} = \overrightarrow{\varepsilon} . \overrightarrow{\alpha} = \overrightarrow{\alpha}$, for all $\overrightarrow{\alpha}$. We denote the $i$-th term of the sequence $\overrightarrow{\alpha}$ by $(\overrightarrow{\alpha})_i$.
\end{defn}

\begin{defn}
We say that a finite sequence of actions $\overrightarrow{\beta}$ is a prefix of $\overrightarrow{\alpha}$ if there is a non-empty sequence $\overrightarrow{\lambda}$ such that $\overrightarrow{\alpha} = \overrightarrow{\beta} . \overrightarrow{\lambda}$. If $\overrightarrow{\beta}$ is a prefix of $\overrightarrow{\alpha}$, we write $\overrightarrow{\beta} \subset \overrightarrow{\alpha}$.
\end{defn}

\begin{defn}
We write $P \stackrel{\overrightarrow{\alpha}}{\Rightarrow} P'$ to express that the process $P$ may perform the sequence of actions $\overrightarrow{\alpha}$ and after that behave as $P'$. We write $P \stackrel{\overrightarrow{\alpha}}{\Rightarrow} \surd$ to express that the process $P$ may successfully finish after performing the sequence of actions $\overrightarrow{\alpha}$ (this, in particular, implies that $\overrightarrow{\alpha}$ is finite).
\end{defn}

\begin{defn}
We define the set of finite possible runs of a process $P$, denoted by $\overrightarrow{{\cal R}_f}(P)$, as $\overrightarrow{{\cal R}_f}(P) = \{\overrightarrow{\alpha} : P \stackrel{\overrightarrow{\alpha}}{\Rightarrow} \surd \}$.
\end{defn}

We want to define semantics for our logics that only take into account the \emph{finite} possible runs of the processes, i.e., situations in which the processes successfully finish. Thus, we present some useful results about finite possible runs.

\begin{defn}
Let $R$ and $S$ be sets of finite sequences of actions. We can define the following operations on these sets:
\begin{enumerate}
 \item $R \circ S = \{ \overrightarrow{\alpha} . \overrightarrow{\beta} : \overrightarrow{\alpha} \in R \,\, \textrm{and} \,\, \overrightarrow{\beta} \in S \}$;
 \item $R \cup S = \{ \overrightarrow{\alpha} : \overrightarrow{\alpha} \in R \,\, \textrm{or} \,\, \overrightarrow{\alpha} \in S \}$;
 \item $R^0 = \{ \overrightarrow{\varepsilon} \}$, $R^n = R \circ R^{n-1} (n \geq 1)$;
 \item $R^* = \bigcup_{n \in \mathbb{N}} R^n$.
\end{enumerate}
\end{defn}

\begin{lem}\label{lem:bisseq}
If $P \sim Q$, then $P \stackrel{\overrightarrow{\alpha}}{\Rightarrow} \surd$ if and only if $Q \stackrel{\overrightarrow{\alpha}}{\Rightarrow} \surd$.
\end{lem}

\begin{proof}
We prove this by induction on the length $n$ of $\overrightarrow{\alpha}$. If $n = 0$, then $\overrightarrow{\alpha} = \overrightarrow{\varepsilon}$ and neither $P$ nor $Q$ may successfully finish without executing any action. If $n = 1$, then $\overrightarrow{\alpha} = \alpha$, for some action $\alpha$. Then, $P \stackrel{\overrightarrow{\alpha}}{\Rightarrow} \surd \Leftrightarrow P \stackrel{\alpha}{\rightarrow} \surd$. By the hypothesis that $P \sim Q$, $P \stackrel{\alpha}{\rightarrow} \surd \Leftrightarrow Q \stackrel{\alpha}{\rightarrow} \surd$. Finally, $Q \stackrel{\alpha}{\rightarrow} \surd \Leftrightarrow Q \stackrel{\overrightarrow{\alpha}}{\Rightarrow} \surd$.

Suppose that the theorem is true for all $n < k$. Let $\overrightarrow{\alpha}$ be a sequence of length $k$. Let $\alpha$ be the first action of the sequence and let $\overrightarrow{\beta}$ be a sequence of length $k-1$ such that $\overrightarrow{\alpha} = \alpha . \overrightarrow{\beta}$. Then, $P \stackrel{\overrightarrow{\alpha}}{\Rightarrow} \surd$ if and only if there is a process $P'$ such that $P \stackrel{\alpha}{\rightarrow} P'$ and $P' \stackrel{\overrightarrow{\beta}}{\Rightarrow} \surd$. But if $P \stackrel{\alpha}{\rightarrow} P'$ and $P \sim Q$, then there is a process $Q'$ such that $Q \stackrel{\alpha}{\rightarrow} Q'$ and $P' \sim Q'$. Now, $\overrightarrow{\beta}$ is a sequence of length shorter than $k$, so by the induction hypothesis, as $P' \sim Q'$ and $P' \stackrel{\overrightarrow{\beta}}{\Rightarrow} \surd$, then $Q' \stackrel{\overrightarrow{\beta}}{\Rightarrow} \surd$. This means that $Q \stackrel{\overrightarrow{\alpha}}{\Rightarrow} \surd$, proving the theorem.
\end{proof}

\begin{thm}\label{teo:eqRf}
If $P \sim Q$, then $\overrightarrow{{\cal R}_f}(P) = \overrightarrow{{\cal R}_f}(Q)$.
\end{thm}

\begin{proof}
Suppose that $\overrightarrow{\alpha} \in \overrightarrow{{\cal R}_f}(P)$. Then, $P \stackrel{\overrightarrow{\alpha}}{\Rightarrow} \surd$. As $P \sim Q$, this implies, by lemma \ref{lem:bisseq}, that $Q \stackrel{\overrightarrow{\alpha}}{\Rightarrow} \surd$, which means that $\overrightarrow{\alpha} \in \overrightarrow{{\cal R}_f}(Q)$. Thus, $\overrightarrow{{\cal R}_f}(P) \subseteq \overrightarrow{{\cal R}_f}(Q)$. The proof that $\overrightarrow{{\cal R}_f}(Q) \subseteq \overrightarrow{{\cal R}_f}(P)$ is entirely analogous.
\end{proof}

\section{sCCS-PDL}\label{sec:sccs}

This section presents our first CCS-Based Dynamic Logic. In this logic, all the CCS processes that appear do not use constants or restriction. We call this logic Small CCS-PDL or sCCS-PDL. Our goal here is to introduce a simple logic and discuss some of the issues concerning the axioms and the relational interpretation of the formulas.

\subsection{Language and Semantics}

In this section, we present the syntax and semantics of sCCS-PDL.

\begin{defn}\label{def-langl} 
The sCCS-PDL language consists of a set $\Phi$ of countably many proposition symbols, a set ${\cal N}$ of countably many names, the silent action $\tau$, the boolean connectives $\neg$ and $\land$, the CCS operators $.$, $+$ and $\mid$ and a modality $\langle P \rangle$ for every process $P$. The formulas are defined as follows:
\[
\varphi ::= p \mid \top \mid \neg \varphi \mid \varphi_1 \land \varphi_2 \mid \langle P \rangle \varphi, \mbox{ with }
\, P ::= \alpha \mid  \alpha . P \mid P_1 + P_2 \mid P_1|P_2,
\]
where $p \in \Phi$ and $\alpha \in {\cal N} \cup \overline{\cal N} \cup \{ \tau \}$.
\end{defn} 

\begin{defn}\label{def-frame1} 
A \emph{frame} for sCCS-PDL is a tuple $\mathcal{F}= (W, \{R_{\alpha}\})$ where
\begin{itemize}
\item $W$ is a non-empty set of states;
\item $R_{\alpha}$ is a binary relation for each basic action $\alpha \in {\cal N} \cup \overline{\cal N} \cup \{ \tau \}$.
\end{itemize}
\end{defn}

\begin{defn}\label{def-model*} 
A \emph{model} for sCCS-PDL is a pair $\mathcal{M}= ({\cal F}, {\bf V})$, where ${\cal F}$ is a sCCS-PDL frame and ${\bf V}$  is a valuation function ${\bf V} : \Phi \mapsto 2^W$.
\end{defn}

We now define the semantical notion of satisfaction for sCCS-PDL as follows:

\begin{defn}\label{def-sat*} 
Let $\mathcal{M}= ({\cal F}, {\bf V})$ be a model. The notion of \emph{satisfaction} of a formula $\varphi$ in a model $\mathcal{M}$ at a state $w$, notation $\mathcal{M},w \Vdash \varphi$, can be inductively defined as follows:
\begin{itemize}
\item $\mathcal{M},w \Vdash p$ iff $w \in {\bf V}(p)$;
\item $\mathcal{M},w \Vdash \top$ always;
\item $\mathcal{M},w \Vdash \neg \varphi$ iff $\mathcal{M},w \not\Vdash \varphi$;
\item $\mathcal{M},w \Vdash \varphi_{1} \wedge \varphi_{2}$ iff $\mathcal{M},w \Vdash \varphi_{1}$ and $\mathcal{M},w \Vdash \varphi_{2}$;
\item $\mathcal{M},w \Vdash \langle P \rangle \varphi$ iff there is a finite path $(v_0,v_1,\ldots,$ $v_n)$, $n \geq 1$, such that $v_0 = w$, $\mathcal{M},v_n \Vdash \varphi$ and there is $\overrightarrow{\alpha} \in \overrightarrow{{\cal R}_f}(P)$ of length $n$ such that $(v_{i-1},v_{i}) \in R_{\beta}$ if and only if $(\overrightarrow{\alpha})_i = \beta$, for $1 \leq i \leq n$. We say that such $\overrightarrow{\alpha}$ \emph{matches} the path $(v_0,\ldots,v_n)$.
\end{itemize}
\end{defn}

If $\mathcal{M},w \Vdash \varphi$ for every state $w$, we say that $\varphi$ is \emph{globally satisfied} in the model $\mathcal{M}$, notation  $\mathcal{M} \Vdash \varphi$. If $\varphi$ is globally satisfied in all models $\mathcal{M}$ of a frame ${\cal F}$, we say that $\varphi$ is \emph{valid} in ${\cal F}$, notation ${\cal F} \Vdash \varphi$. Finally, if $\varphi$ is valid in all frames, we say that $\varphi$ is valid, notation $\Vdash \varphi$. Two formulas $\varphi$ and $\psi$ are \emph{semantically equivalent} if $\Vdash \varphi \leftrightarrow \psi$.

As mentioned in the previous section, there are two possible semantics for the $\tau$ action in CCS: it can be regarded as being observable or 
as being invisible. In our logics, we adopt the first one, since we are able to represent the second semantics as a particular case of the first. In fact, to do that, the only thing that is necessary is to force, in the frames under consideration, $R_{\tau}$ to be the relation $R_{\tau} = \{ (w,w) : w \in W \}$.

\begin{thm}\label{teo:valPQ}
$\overrightarrow{{\cal R}_f}(P) = \overrightarrow{{\cal R}_f}(Q)$ if and only if $\Vdash \langle P \rangle p \leftrightarrow \langle Q \rangle p$.
\end{thm}

\begin{proof}
($\Rightarrow$) Suppose that $\overrightarrow{{\cal R}_f}(P) = \overrightarrow{{\cal R}_f}(Q)$, but $\not\Vdash \langle P \rangle p \leftrightarrow \langle Q \rangle p$. Then, we may assume, without loss of generality, that there is a model ${\cal M}$ and a state $v_0$ in this model such that ${\cal M}, v_0 \Vdash \langle P \rangle p$ (*), but ${\cal M}, v_0 \not\Vdash \langle Q \rangle p$ (**). By definition \ref{def-sat*}, (*) implies that there is a path $(v_0,v_1,\ldots,v_n)$, $n \geq 1$, in ${\cal M}$ such that ${\cal M}, v_n \Vdash p$ (***) and there is $\overrightarrow{\alpha} \in \overrightarrow{{\cal R}_f}(P)$ that matches this path. But as $\overrightarrow{{\cal R}_f}(P) = \overrightarrow{{\cal R}_f}(Q)$, then $\overrightarrow{\alpha} \in \overrightarrow{{\cal R}_f}(Q)$. This and (***) imply, by definition \ref{def-sat*}, that ${\cal M}, v_0 \Vdash \langle Q \rangle p$, contradicting (**).

($\Leftarrow$) Suppose that $\Vdash \langle P \rangle p \leftrightarrow \langle Q \rangle p$ (*), but $\overrightarrow{{\cal R}_f}(P) \neq \overrightarrow{{\cal R}_f}(Q)$. Then, we may assume, without loss of generality, that there is $\overrightarrow{\alpha}$ such that $\overrightarrow{\alpha} \in \overrightarrow{{\cal R}_f}(P)$, but $\overrightarrow{\alpha} \not\in \overrightarrow{{\cal R}_f}(Q)$. Let us build a frame ${\cal F}$ that consists solely of a path $(v_0,\ldots,v_n)$, $n \geq 1$, such that $R_{\alpha} = \{(v_{i-1},v_i) : 1 \leq i \leq n \,\, \textrm{and} \,\, \alpha \,\, \textrm{is the i-th term of} \,\, \overrightarrow{\alpha} \}$. Let ${\cal M} = ({\cal F}, {\bf V})$, such that $v_n \in {\bf V}(p)$ and $v_i \not\in {\bf V}(p)$, $1 \leq i < n$. Then, we have a path $(v_0,\ldots,v_n)$ such that ${\cal M}, v_n \Vdash p$ and $\overrightarrow{\alpha} \in \overrightarrow{{\cal R}_f}(P)$ matches this path. By definition \ref{def-sat*}, ${\cal M}, v_0 \Vdash \langle P \rangle p$. However, $\overrightarrow{\alpha} \not\in \overrightarrow{{\cal R}_f}(Q)$, so $(v_0,\ldots,v_n)$ is not matched by any sequence in $\overrightarrow{{\cal R}_f}(Q)$. Besides that, there is no other path $(v_0,\ldots,v_m)$, $m \geq 1$, in ${\cal M}$ such that ${\cal M}, v_m \Vdash p$. Thus, by definition \ref{def-sat*}, ${\cal M}, v_0 \not\Vdash \langle Q \rangle p$, which contradicts (*).
\end{proof}

\begin{cor}\label{cor:simPQvalPQ}
If $P \sim Q$, then $\Vdash \langle P \rangle p \leftrightarrow \langle Q \rangle p$.
\end{cor}

\begin{proof}
It follows directly from theorems \ref{teo:eqRf} and \ref{teo:valPQ}.
\end{proof}

We present some equalities between sets of finite possible runs that are useful to the soundness proof of our axiomatization and to show why the null process ${\bf 0}$ is problematic.

\begin{thm}\label{teo:Pf}
The following set equalities are true:
\begin{enumerate}
 \item $\overrightarrow{{\cal R}_f}(\alpha) = \{ \alpha \}$;
 \item $\overrightarrow{{\cal R}_f}(\alpha . P) = \overrightarrow{{\cal R}_f}(\alpha) \circ \overrightarrow{{\cal R}_f}(P)$;
 \item $\overrightarrow{{\cal R}_f}(P_1 + P_2) = \overrightarrow{{\cal R}_f}(P_1) \cup \overrightarrow{{\cal R}_f}(P_2)$.
\end{enumerate}
\end{thm}

\begin{proof}
The proof is straightforward from table \ref{tab:semccs}.
\end{proof}

\begin{thm}\label{teo:valid}
The following formulas are valid:
\begin{enumerate}
 \item $\langle \alpha . P \rangle p \leftrightarrow \langle \alpha \rangle \langle P \rangle p$
 \item $\langle P_{1} + P_{2} \rangle p \leftrightarrow \langle P_{1} \rangle p \lor \langle P_{2} \rangle p$
\end{enumerate}
\end{thm}

\begin{proof}
We only provide the proof for the first formula. The proof for the second formula follows by an analogous line of reasoning, using the third equality in theorem \ref{teo:Pf} instead of the second one.

($\Rightarrow$) Suppose that, for some model $\mathcal{M}$ and some state $w$ in this model, $\mathcal{M},w \Vdash \langle \alpha . P \rangle p$. Then, by definition \ref{def-sat*}, there is a finite path $(v_0,v_1,\ldots,v_n)$, $n \geq 1$, such that $v_0 = w$, $\mathcal{M},v_n \Vdash p$ and a sequence $\overrightarrow{\alpha} \in \overrightarrow{{\cal R}_f}(\alpha . P)$ that matches this path. Now, by the first and second equalities in theorem \ref{teo:Pf}, there is a sequence $\overrightarrow{\beta} \in \overrightarrow{{\cal R}_f}(P)$ such that $\overrightarrow{\alpha} = \alpha . \overrightarrow{\beta}$. $\overrightarrow{\beta}$ matches the path $(v_1,\ldots,v_n)$, which implies that $\mathcal{M},v_1 \Vdash \langle P \rangle p$. Besides that, $\alpha$ matches the path $(v_0,v_1)$, which implies that $\mathcal{M},w \Vdash \langle \alpha \rangle \langle P \rangle p$. Thus, $\langle \alpha . P \rangle p \rightarrow \langle \alpha \rangle \langle P \rangle p$ is valid.

($\Leftarrow$) This proof is entirely analogous to the previous one, using the second equality in theorem \ref{teo:Pf} in the reverse direction.
\end{proof}

Now it is possible to see why, as stated in the previous section, the use of the null process ${\bf 0}$ in our logics would be inconvenient. The problems that would appear come from the fact that, as described in \cite{M89}, in a specification of the form $\alpha . {\bf 0}$, ${\bf 0}$ is denoting a process that has successfully terminated, while in a specification of the form $P + {\bf 0}$, ${\bf 0}$ is denoting a deadlocked process. This double role cannot be kept in our logics without sacrificing a very desirable property in a dynamic logic: the compositional semantics, illustrated in theorem \ref{teo:valid}.

The compositional semantics is a direct consequence of the set equalities in theorem \ref{teo:Pf}. But when we try to keep them in the presence of ${\bf 0}$, some problems arise. $\overrightarrow{{\cal R}_f}(\alpha . {\bf 0}) = \{ \alpha \}$, since {\bf 0} denotes successful termination in this case (if {\bf 0} denoted a deadlock, then $\overrightarrow{{\cal R}_f}(\alpha . {\bf 0})$ would be $\emptyset$), and $\overrightarrow{{\cal R}_f}(P + {\bf 0}) = \overrightarrow{{\cal R}_f}(P)$, since {\bf 0} denotes a deadlock in this case (if {\bf 0} denoted successful termination, then $\overrightarrow{{\cal R}_f}(P + {\bf 0})$ would be $\overrightarrow{{\cal R}_f}(P) \cup \{ \overrightarrow{\varepsilon} \}$). To keep the second equality in theorem \ref{teo:Pf}, we must have $\{ \alpha \} = \overrightarrow{{\cal R}_f}(\alpha . {\bf 0}) = \overrightarrow{{\cal R}_f}(\alpha) \circ \overrightarrow{{\cal R}_f}({\bf 0})$, which implies that $\overrightarrow{{\cal R}_f}({\bf 0}) = \{ \overrightarrow{\varepsilon} \}$ (*). On the other hand, to keep the third equality, we must have $\overrightarrow{{\cal R}_f}(P) = \overrightarrow{{\cal R}_f}(P + {\bf 0}) = \overrightarrow{{\cal R}_f}(P) \cup \overrightarrow{{\cal R}_f}({\bf 0})$, which implies that $\overrightarrow{{\cal R}_f}({\bf 0}) = \emptyset$ (**).

In the logical formalism, by theorem \ref{teo:valid}, (*) would imply that $\langle {\bf 0} \rangle \phi$ is semantically equivalent to $\phi$, while (**) would imply that $\langle {\bf 0} \rangle \phi$ is semantically equivalent to $\bot$. The crucial point in this situation is that we would have to either abandon at least one of the equalities in theorem \ref{teo:Pf}, substituting it by a pair of equations, one for the case where $P \neq {\bf 0}$ and the other for the case where $P = {\bf 0}$, or to somehow change the semantics so that the meaning of a subformula of the form $\langle {\bf 0} \rangle \phi$ will depend on the context in which it is inserted, being sometimes equivalent to $\phi$ and sometimes to $\bot$. Both ``solutions'' would seriously compromise the compositionality of the semantics.

We address this issue of the null process in our third logic, without introducing any of the above problems. There, we redefine the process ${\bf 0}$ so that it denotes only a deadlocked process, while defining a new way to denote termination.

\subsection{Axiomatic System}\label{pt}

We consider the following set of axioms and rules, where $p$ and $q$ are proposition symbols and $\varphi$ and $\psi$ are formulas.

\begin{description}
\item[(PL)] Enough propositional logic tautologies
\item[(K)] $\vdash [P](p \rightarrow q) \rightarrow ([P]p \rightarrow [P]q)$
\item[(Du)] $\vdash [P] p \leftrightarrow \neg \langle P \rangle \neg p$
\item[(Pr)] $\vdash \langle \alpha . P \rangle p \leftrightarrow \langle \alpha \rangle \langle P \rangle p$
\item[(NC)] $\vdash \langle P_{1} + P_{2} \rangle p \leftrightarrow \langle P_{1} \rangle p \lor \langle P_{2} \rangle p$
\item[(PC)] If EL can be applied to $P$, then $\vdash \langle P \rangle p \leftrightarrow \langle Exp(P) \rangle p$
\item[(Sub)] If $\vdash \varphi$, then $\vdash \varphi^\sigma$, where $\sigma$ uniformly substitutes proposition symbols by arbitrary formulas.
\item[(MP)] If $\vdash \varphi$ and $\vdash \varphi \rightarrow \psi$, then $\vdash \psi$.
\item[(Gen)] If $\vdash \varphi$, then $\vdash [P]\varphi$.
\end{description}

It is important to notice that the theorems $\vdash \langle P_1 + P_2 \rangle p \leftrightarrow \langle P_2 + P_1 \rangle p$ and $\vdash \langle P_1 | P_2 \rangle p \leftrightarrow \langle P_2 | P_1 \rangle p$, which state the commutativity of the $+$ and $|$ operators, are derivable from the axiomatic system above.

The axioms {\bf (PL)}, {\bf (K)} and {\bf (Du)} and the rules {\bf (Sub)}, {\bf (MP)} and {\bf (Gen)} are standard in the modal logic literature. The soundness of {\bf (Pr)} and {\bf (NC)} follows directly from the set equalities in theorem \ref{teo:Pf} and from definition \ref{def-sat*}, as shown in theorem \ref{teo:valid}. Finally, the soundness of {\bf (PC)} follows from theorem \ref{teo:EL} and corollary \ref{cor:simPQvalPQ}.

The above axiomatic system is also complete with respect to the class of sCCS-PDL frames and the logic has the finite model property. We omit the proofs  here, because they are analogous to the proofs presented  in section \ref{pdl-ccs}, where constants are added to the language.

\section{CCS-PDL}\label{pdl-ccs}

The logic presented in this section uses the same CCS operators as in the previous section plus constants. This is the CCS-PDL logic. Our goal in this section is to build an axiomatic system for CCS-PDL and prove its completeness.

\subsection{Language and Semantics}

In this section, we present the syntax and semantics of CCS-PDL.

\begin{defn}\label{def-lang*} 
The CCS-PDL language consists of a set $\Phi$ of countably many proposition symbols, a set ${\cal N}$ of countably many names, the silent action $\tau$, the boolean connectives $\neg$ and $\land$, the CCS operators $.$, $+$ and $\mid$, a set ${\cal C}$ of countably many constants, such that each element of ${\cal C}$ has its unique correspondent defining equation, and a modality $\langle P \rangle$ for every process $P$. The formulas are defined as follows:
\[
\varphi ::= p \mid \top \mid \neg \varphi \mid \varphi_1 \land \varphi_2 \mid \langle P \rangle \varphi, \mbox{ with }
\, P ::= \alpha \mid  \alpha . P \mid \alpha . A \mid P_1 + P_2 \mid P_1|P_2,
\]
where $p \in \Phi$, $\alpha \in {\cal N} \cup \overline{\cal N} \cup \{ \tau \}$ and $A \in {\cal C}$.
\end{defn}

The presence of constants in the language allows us to write specifications that are capable of iteration, as $P = \alpha . A$, with $A \stackrel{def}{=} \alpha . A + \tau$. However, constants have a much greater power than just expressing iterative behaviours. With constants, we are able to write self-replicating specifications, as $P = ((\tau . A) + \tau) | Q$, with $A \stackrel{def}{=} ((\tau . A) + \tau) | Q$. After the execution of $n$ $\tau$-actions, $P$ is capable of behaving as $n$ $Q$-processes in parallel, for any $n \in \mathbb{N}$.

The example above is a very simple example of self-replication and it is easy to see that things can get very complex if we start nesting self-replicating processes. 

In order to keep the logic simple, that is, keep the simple Kripke semantics, the finite model property and a simple and complete axiomatization, we restrict the use of constants in CCS-PDL in order to prevent self-replicating processes (in \cite{Dam}, Dam enforces a similar syntactic restriction, also to prevent unbounded process growth). The issue of whether it is possible to keep these desirable properties of the logic in the presence of replication remains an open problem and we defer it to a future work, as explained in section \ref{DC}.

\begin{defn}
Let $P$ be a process and $\{A_1,\ldots,A_n\}$ be the constants that occur in $P$. We define $Cons(P)$ as the smallest set of constants such that $Cons(P) \supseteq \{A_1,\ldots,A_n\}$ and, for every constant $A_i \in Cons(P)$, if $A_k$ occurs in $P_{A_i}$, then $A_k \in Cons(P)$.
\end{defn}

\begin{res}
We make the following restrictions to processes in CCS-PDL:
\begin{enumerate}
 \item $Cons(P)$ must be a finite set for every process $P$;
 \item We only allow defining equations that fit into one of the following models:
\begin{itemize}
 \item $A \stackrel{def}{=} P_A$, where $A \notin Cons(P_A)$, called \emph{non-recursive equations};
 \item $A \stackrel{def}{=} \overrightarrow{\alpha}_1 . A + \ldots + \overrightarrow{\alpha}_n . A + T_A$, where $A \notin Cons(T_A)$, called \emph{recursive equations}.
\end{itemize}
\end{enumerate}
\end{res}

The set equalities from theorem \ref{teo:Pf} remains valid, along with the equality 
\begin{equation}\label{eq:const}
\overrightarrow{{\cal R}_f}(\alpha . A) = \overrightarrow{{\cal R}_f}(\alpha) \circ \overrightarrow{{\cal R}_f}(P_A),
\end{equation}
which also follows from table \ref{tab:semccs}. 

However, due to the possibility of iterative behaviours, some set equalities may present themselves as recursive equations. In these cases, it is possible to obtain an equivalent non-recursive equality. First, the recursive equation can be rewritten, using the set equalities in theorem \ref{teo:Pf} and equation \eqref{eq:const}, as $\overrightarrow{{\cal R}_f}(P) = \overrightarrow{{\cal R}_f}(P') \circ \overrightarrow{{\cal R}_f}(P) \cup \overrightarrow{{\cal R}_f}(Q)$, where $\overrightarrow{{\cal R}_f}(Q)$ is not a function of $\overrightarrow{{\cal R}_f}(P)$. Now, as all sequences in $\overrightarrow{{\cal R}_f}(P)$, $\overrightarrow{{\cal R}_f}(P')$ and $\overrightarrow{{\cal R}_f}(Q)$ are finite and $\overrightarrow{\varepsilon} \not\in \overrightarrow{{\cal R}_f}(P')$, we may use a result known as Arden's Rule, that states that if $X$, $A$ and $B$ are sets of finite strings and the empty string is not in $A$, then the equation $X = A \circ X \cup B$ has as its unique solution $X = A^* \circ B$ \cite{Arden}. Thus, $\overrightarrow{{\cal R}_f}(P) = \overrightarrow{{\cal R}_f}^*(P') \circ \overrightarrow{{\cal R}_f}(Q)$.

\begin{defn}
We say that a process $P$ is a \emph{knot process} if $P = P_A$ for some constant $A$ with a recursive defining equation or if $P = P_1 \mid P_2$ where $P_1$ or $P_2$ is a knot process. Otherwise, we say that $P$ is a \emph{non-knot process}.
\end{defn}

\begin{defn}\label{def:loop}
We call a non-empty sequence of actions $\overrightarrow{\alpha}$ a \emph{loop} of a knot process $P$ if $P \stackrel{\overrightarrow{\alpha}}{\Rightarrow} P$. We say that $\overrightarrow{\alpha}$ is a \emph{proper loop} if $\overrightarrow{\alpha}$ is a loop and there is no $\overrightarrow{\beta} \subset \overrightarrow{\alpha}$, with $\overrightarrow{\alpha} = \overrightarrow{\beta} . \overrightarrow{\lambda}$, such that $\overrightarrow{\beta}$ and $\overrightarrow{\lambda}$ are loops of $P$. The set of loops of $P$ is denoted by $Lo(P)$ and the set of proper loops of $P$ is denoted by $PLo(P)$.
\end{defn}

\begin{thm}\label{teo:loop}
$\overrightarrow{\alpha} \in Lo(P)$ if and only if $\overrightarrow{\alpha} = \overrightarrow{\alpha_1}. \cdots . \overrightarrow{\alpha_n}$, $n \geq 1$, where $\overrightarrow{\alpha_i} \in PLo(P)$, for all $i \in \{1,\ldots,n\}$.
\end{thm}

\begin{proof}
The proof is straightforward from definition \ref{def:loop}.
\end{proof}

\begin{defn}\label{def:break}
We call a sequence of actions $\overrightarrow{\alpha}$ a \emph{breaker} of a knot process $P$ if there is no $\overrightarrow{\beta}$ such that $\overrightarrow{\alpha} \subset \overrightarrow{\beta}$ and $\overrightarrow{\beta}$ is a loop. We say that $\overrightarrow{\alpha}$ is a \emph{proper breaker} if $\overrightarrow{\alpha}$ is a breaker and there is no $\overrightarrow{\beta} \subset \overrightarrow{\alpha}$, with $\overrightarrow{\alpha} = \overrightarrow{\beta} . \overrightarrow{\lambda}$, such that $\overrightarrow{\beta}$ is a loop and $\overrightarrow{\lambda}$ is a breaker. Finally, we say that $\overrightarrow{\alpha}$ is a \emph{minimal proper breaker} if $\overrightarrow{\alpha}$ is a proper breaker and there is no $\overrightarrow{\beta} \subset \overrightarrow{\alpha}$ such that $\overrightarrow{\beta}$ is a proper breaker. The set of breakers of $P$ is denoted by $Br(P)$, the set of proper breakers of $P$ is denoted by $PBr(P)$ and the set of minimal proper breakers of $P$ is denoted by $MPBr(P)$.
\end{defn}

\begin{thm}\label{teo:break}
$\overrightarrow{\alpha} \in PBr(P)$ if and only if $\overrightarrow{\alpha} = \overrightarrow{\beta} . \overrightarrow{\lambda}$, where $\overrightarrow{\beta} \in MPBr(P)$.
\end{thm}

\begin{proof}
The proof is straightforward from definition \ref{def:break}.
\end{proof}

Using the concepts of loops and breakers, we can split a knot process $P$ into two parts: the \emph{looping part}, denoted by $L_P$, and the \emph{tail part}, denoted by $T_P$.
\[
L_P = \sum \{ \overrightarrow{\alpha} : \overrightarrow{\alpha} \in PLo(P) \}.
\]
and
\[
T_P = \sum \{ \overrightarrow{\alpha}.P' : \overrightarrow{\alpha} \in MPBr(P) \,\, \textrm{and} \,\, P \stackrel{\overrightarrow{\alpha}}{\Rightarrow} P' \}
\]

\begin{thm}\label{teo:decknot}
If $P$ is a knot process, then $\overrightarrow{{\cal R}_f}(P) = \overrightarrow{{\cal R}_f}^*(L_P) \circ \overrightarrow{{\cal R}_f}(T_P)$.
\end{thm}

\begin{proof}
We show that $\overrightarrow{{\cal R}_f}(P) = \overrightarrow{{\cal R}_f}(L_P) \circ \overrightarrow{{\cal R}_f}(P) \cup \overrightarrow{{\cal R}_f}(T_P)$. The result then follows from Arden's Rule \cite{Arden}, since $\overrightarrow{\varepsilon} \not\in \overrightarrow{{\cal R}_f}(L_P)$.

If $\overrightarrow{\alpha} \in \overrightarrow{{\cal R}_f}(L_P) \circ \overrightarrow{{\cal R}_f}(P)$, then $\overrightarrow{\alpha} = \overrightarrow{\beta} . \overrightarrow{\lambda}$, where $\overrightarrow{\beta} \in \overrightarrow{{\cal R}_f}(L_P)$ and $\overrightarrow{\lambda} \in \overrightarrow{{\cal R}_f}(P)$. Then, $P \stackrel{\overrightarrow{\beta}}{\Rightarrow} P$ and $P \stackrel{\overrightarrow{\lambda}}{\Rightarrow} \surd$, which implies that $P \stackrel{\overrightarrow{\alpha}}{\Rightarrow} \surd$. Thus, $\overrightarrow{\alpha} \in \overrightarrow{{\cal R}_f}(P)$. If $\overrightarrow{\alpha} \in \overrightarrow{{\cal R}_f}(T_P)$, then $\overrightarrow{\alpha} = \overrightarrow{\beta} . \overrightarrow{\lambda}$, where $P \stackrel{\overrightarrow{\beta}}{\Rightarrow} P'$ and $P' \stackrel{\overrightarrow{\lambda}}{\Rightarrow} \surd$, which implies that $P \stackrel{\overrightarrow{\alpha}}{\Rightarrow} \surd$. Thus, $\overrightarrow{\alpha} \in \overrightarrow{{\cal R}_f}(P)$. This proves that $\overrightarrow{{\cal R}_f}(L_P) \circ \overrightarrow{{\cal R}_f}(P) \cup \overrightarrow{{\cal R}_f}(T_P) \subseteq \overrightarrow{{\cal R}_f}(P)$.

If $\overrightarrow{\alpha} \in \overrightarrow{{\cal R}_f}(P)$, then we have two cases:
\begin{enumerate}
 \item There is $\overrightarrow{\beta} \subset \overrightarrow{\alpha}$, with $\overrightarrow{\alpha} = \overrightarrow{\beta} . \overrightarrow{\lambda}$, such that $\overrightarrow{\beta}$ is a loop. Then, by theorem \ref{teo:loop}, $\overrightarrow{\beta} = \overrightarrow{\beta_1} . \overrightarrow{\beta_2}$, where $\overrightarrow{\beta_1} \in PLo(P)$. This means that $\overrightarrow{\beta_1} \in \overrightarrow{{\cal R}_f}(L_P)$.
If we make $\overrightarrow{\gamma} = \overrightarrow{\beta_2} . \overrightarrow{\lambda}$, then $\overrightarrow{\alpha} = \overrightarrow{\beta_1} . \overrightarrow{\gamma}$ and $P \stackrel{\overrightarrow{\gamma}}{\Rightarrow} \surd$. Thus, $\overrightarrow{\gamma} \in \overrightarrow{{\cal R}_f}(P)$ and $\overrightarrow{\alpha} \in \overrightarrow{{\cal R}_f}(L_P) \circ \overrightarrow{{\cal R}_f}(P)$.
 \item There is no $\overrightarrow{\beta} \subset \overrightarrow{\alpha}$, with $\overrightarrow{\alpha} = \overrightarrow{\beta} . \overrightarrow{\lambda}$, such that $\overrightarrow{\beta}$ is a loop. This implies that, for all $\overrightarrow{\beta} \subset \overrightarrow{\alpha}$, $\overrightarrow{\beta} \in PBr(P)$. Then, by theorem
\ref{teo:break}, $\overrightarrow{\beta} = \overrightarrow{\beta_1} . \overrightarrow{\beta_2}$, where $\overrightarrow{\beta_1} \in MPBr(P)$. If we make $\overrightarrow{\gamma} = \overrightarrow{\beta_2} . \overrightarrow{\lambda}$, then $\overrightarrow{\alpha} = \overrightarrow{\beta_1} . \overrightarrow{\gamma}$. This means that, if $P \stackrel{\overrightarrow{\beta_1}}{\Rightarrow} P'$, then $P' \stackrel{\overrightarrow{\gamma}}{\Rightarrow} \surd$. Thus, $\overrightarrow{\alpha} \in \overrightarrow{{\cal R}_f}(T_P)$.
\end{enumerate}
This proves that $\overrightarrow{{\cal R}_f}(P) \subseteq \overrightarrow{{\cal R}_f}(L_P) \circ \overrightarrow{{\cal R}_f}(P) \cup \overrightarrow{{\cal R}_f}(T_P)$.
\end{proof}

We also define the process $L_P'$, that is capable of iterating $L_P$.
\[
L_P' = \sum \{ \overrightarrow{\alpha} . Z_P : \overrightarrow{\alpha} \in PLo(P) \} + L_P,
\]
where $Z_P$ is a new constant with defining equation $Z_P \stackrel{def}{=} L_P'$. 

The notions of frame, model and satisfaction are defined analogously to definitions \ref{def-frame1}, \ref{def-model*} and \ref{def-sat*}. It is not difficult to see that theorem \ref{teo:valPQ} and corollary \ref{cor:simPQvalPQ} remain valid in CCS-PDL.

\subsection{Axiomatic System}\label{pt*}

The axiomatic system is similar to the one presented in section \ref{pt}. We consider the following set of axioms and rules, where $p$, $q$ and $r$ are proposition symbols and $\varphi$ and $\psi$ are formulas.

\begin{itemize}
\item The axioms {\bf (PL)}, {\bf (K)} and {\bf (Du)} and the rules {\bf (Sub)}, {\bf (MP)} and {\bf (Gen)}.
\item Axioms for knot processes:
\begin{description}
\item[(Rec)] $\vdash \langle P \rangle p \leftrightarrow \langle T_P \rangle p \lor \langle L_P \rangle \langle P \rangle p$
\item[(FP)] $\vdash (r \rightarrow ([T_P] \neg p \land [L_P] r)) \land [L_P'](r \rightarrow ([T_P] \neg p \land [L_P] r)) \rightarrow (r \rightarrow [P] \neg p)$
\end{description}
\item Axioms for non-knot processes:
\begin{description}
\item[(sCCS)] The axioms {\bf (Pr)}, {\bf (NC)} and the rule {\bf (PC)}.
\item[(Cons)] $\vdash \langle \alpha . A \rangle p \leftrightarrow \langle \alpha \rangle \langle P_A \rangle p$
\end{description}
\end{itemize}

The proof of soundness is analogous to the proof of soundness for sCCS-PDL. The soundness of {\bf (Rec)} and {\bf (FP)} follows from theorems \ref{teo:decknot} and \ref{teo:valPQ}. The axiom {\bf (FP)} may seem strange at first, but it is just an adaptation of the so-called induction axiom to our particular situation. The soundness of {\bf (Cons)} follows from equation \eqref{eq:const} and theorem \ref{teo:valPQ}.

\begin{thm}[Completeness]
Every consistent formula is satisfiable in a finite CCS-PDL model.
\end{thm}

\begin{proof}
The proof is presented in the appendix \ref{sec:compproof}.
\end{proof}

\section{XCCS-PDL}\label{sec:xccs}

As it was shown in section \ref{sec:sccs}, the use of the null process {\bf 0} of CCS in our first two logics would bring a serious inconvenience to their semantics: their compositionality would be compromised. This problem also affect our ability to include the restriction operator in these logics. Besides that, in CCS-PDL, we have to define two distinct sets of axioms, depending on whether the process under consideration is a knot process or not.

In this section, our goal is to solve these two problems that occur in the previous logics. In order to accomplish this, first we extend the language of CCS with new operators and a new type of action and slightly redefine its semantics. We call this new process algebra \emph{extended CCS} or XCCS. Then, we define a dynamic logic in which the programs are XCCS processes (XCCS-PDL). Because of the refined definition of the null process {\bf 0} in XCCS, we can include it in this logic, as well as the restriction operator. Besides that, one of the new operators of XCCS, the iteration operator, allows us to deal with all sorts of processes with just one set of axioms and to also drop the constants and all its elaborated theory from the language.

\subsection{XCCS}

In CCS, we have the set of actions ${\cal A} = {\cal N} \cup {\cal \overline{N}} \cup \{ \tau \}$. In XCCS, we denote this set of actions as ${\cal A}_R$, the set of \emph{running actions}. In XCCS, we have an extra action, besides the ones in ${\cal A}_R$, called the \emph{ending action} and denoted by $END$. A process in XCCS can only successfully finish after performing the action $END$ and it always successfully finishes after performing such action. If a process cannot perform any running action and cannot successfully finish, it is called a \emph{deadlocked} process.

\begin{defn}\label{def:xccslang}
In XCCS, process specifications can be built using the following operations:
\[
P ::=  {\bf 0} \mid END \mid \alpha.P \mid P_1 ; P_2 \mid P_1 + P_2 \mid P_1 | P_2 \mid P^* \mid P \backslash L,
\]
with
\[
\alpha ::= a \mid \overline{a} \mid \tau,
\]
where $a \in {\cal N}$ and $L \subseteq {\cal N}$.
\end{defn}

{\bf 0} is the \emph{null process}. It is a deadlocked process, since it is incapable of performing any running action and of successfully finishing. $END$ is process that is incapable of performing any running action, but it is capable of successfully finishing. The \emph{sequential composition} operator (;) denotes that the process will first behave as $P_1$ and if and when $P_1$ successfully terminates, it will proceed behaving as $P_2$. The \emph{iteration} operator (*) denotes that the process $P$ is capable of being iterated zero or more times. In table \ref{tab:semxccs}, we present the semantics for the XCCS operators.

\begin{table}
\centering
\caption{Transition Relations of XCCS}
\begin{tabular}{|c|c|c|c|c|}
\hline
$\alpha.P \stackrel{\alpha}{\rightarrow} P$ &
$END \stackrel{END}{\rightarrow} \surd$ & 
$P^* \stackrel{END}{\rightarrow} \surd$ &
$\frac{P \stackrel{\alpha}{\rightarrow} P'} {P ; Q \stackrel{\alpha}{\rightarrow} P' ; Q}$ &
$\frac{P \stackrel{END}{\rightarrow} \surd, Q \stackrel{\alpha}{\rightarrow} Q'} {P ; Q \stackrel{\alpha}{\rightarrow} Q'}$ \\
\hline
$\frac{P \stackrel{\alpha}{\rightarrow} P'} {P + Q \stackrel{\alpha}{\rightarrow} P'}$ &  
$\frac{Q \stackrel{\beta}{\rightarrow} Q'} {P + Q \stackrel{\beta}{\rightarrow} Q'}$ &
$\frac{P \stackrel{\alpha}{\rightarrow} P'}{P | Q \stackrel{\alpha}{\rightarrow} P' | Q}$ &
$\frac{Q \stackrel{\beta}{\rightarrow} Q'}{P | Q \stackrel{\beta}{\rightarrow} P | Q'}$ &
$\frac{P \stackrel{\lambda}{\rightarrow} P', Q \stackrel{\overline{\lambda}}{\rightarrow} Q'}{P | Q \stackrel{\tau}{\rightarrow} P' | Q'}$ \\
\hline
$\frac{P \stackrel{\alpha}{\rightarrow} P'} {P^* \stackrel{\alpha}{\rightarrow} P';P^*}$ &
$\frac{P \stackrel{\alpha}{\rightarrow} P',\alpha \not\in L \cup \overline{L}}{P \backslash L \stackrel{\alpha}{\rightarrow} P' \backslash L}$ &
$\frac{P \stackrel{END}{\rightarrow} \surd, Q \stackrel{END}{\rightarrow} \surd}{P;Q \stackrel{END}{\rightarrow} \surd}$ & 
$\frac{P \stackrel{END}{\rightarrow} \surd} {P + Q \stackrel{END}{\rightarrow} \surd}$ &
$\frac{Q \stackrel{END}{\rightarrow} \surd} {P + Q \stackrel{END}{\rightarrow} \surd}$ \\
\hline &
$\frac{P \stackrel{END}{\rightarrow} \surd, Q \stackrel{END}{\rightarrow} \surd}{P | Q \stackrel{END}{\rightarrow} \surd}$ & &
$\frac{P \stackrel{END}{\rightarrow} \surd}{P \backslash L \stackrel{END}{\rightarrow} \surd}$ & \\
\hline
\end{tabular}
\label{tab:semxccs}
\end{table}

From table \ref{tab:semxccs}, it is not difficult to see that now the null process ${\bf 0}$ denotes only a deadlocked process. As explained in section \ref{sec:sccs}, the situation in standard CCS is different, since there, in a specification of the form $\alpha . {\bf 0}$, ${\bf 0}$ is denoting a process that has successfully terminated. This is no longer the case. In XCCS, a specification of the form $\alpha . {\bf 0}$ denotes that a process performs the action $\alpha$ and then deadlocks, while a specification of the form $\alpha . END$ denotes that a process performs the action $\alpha$ and then successfully terminates. This slight extension of the language allows for the null process and for the restriction operator to be used in our third logic without compromising the compositionality of its semantics.

In \cite{M89} and \cite{M99}, Milner uses a clever syntactic construction to define a form of sequential composition. It is slightly different to the form presented in table \ref{tab:semxccs} and it is not a primitive operator. He also uses the notation $;$ for it, but we denote his construction with a $:$ instead, so we can easily differentiate between his and our constructions. Milner's construction depends on a number of things. First, we must consider a new name $z \not\in {\cal N}$. Second, every process must perform the action $\overline{z}$ as its last action before termination and may not perform $z$ or $\overline{z}$ at any other point of execution. Third, we must perform syntactic substitutions of names in processes, where $P[b/a]$ denotes the substitution of every occurrence of $a$ ($\overline{a}$) in $P$ by $b$ ($\overline{b}$). Then, sequential composition is defined in the following way:
\[
P:Q = (P[a/z] \mid a . Q) \backslash \{a\},
\]
where $a$ must be a name that does not occur in neither $P$ nor $Q$.

The main difference between the two forms of sequential composition is that, as tables \ref{tab:semccs} and \ref{tab:semxccs} easily show, $\overrightarrow{{\cal R}_f}(P;Q) = \overrightarrow{{\cal R}_f}(P) \circ \overrightarrow{{\cal R}_f}(Q)$, while $\overrightarrow{{\cal R}_f}(P:Q) = \overrightarrow{{\cal R}_f}(P) \circ \{ \tau \} \circ \overrightarrow{{\cal R}_f}(Q)$. The extra $\tau$ would also be present in the finite runs of a process $P^{*}$, as we use sequential composition to define the semantics of the iteration operator (table \ref{tab:semxccs}). These extra $\tau$'s appearing between the finite runs of the subprocesses would be a complication to the semantics of our logic, as some intuitive validities, such as $\langle A \rangle \langle B \rangle \varphi \rightarrow \langle A;B \rangle \varphi$, would be false. Since we are already introducing the $END$ process to solve the previous problems with the null process, there is no reason why we should not also use it to build a simpler and more convenient form of sequential composition and a simple form of iteration, as it is done in table \ref{tab:semxccs}.

Now, we need to make slight adjustments to the notion of strong bisimulation and to the Expansion Law.

\begin{defn}
Let ${\cal P}$ be the set of all possible process specifications. A set $Z \subseteq {\cal P} \times {\cal P}$ is a \emph{strong bisimulation} if $(P,Q) \in Z$ implies, for all $\alpha \in {\cal A}_R$,
\begin{itemize}
 \item If $P \stackrel{\alpha}{\rightarrow} P'$ and $P' \in {\cal P}$, then there is $Q' \in {\cal P}$ such that $Q \stackrel{\alpha}{\rightarrow} Q'$ and $(P',Q') \in Z$;
 \item If $Q \stackrel{\alpha}{\rightarrow} Q'$ and $Q' \in {\cal P}$, then there is $P' \in {\cal P}$ such that $P \stackrel{\alpha}{\rightarrow} P'$ and $(P',Q') \in Z$;
\end{itemize}
and
\begin{itemize}
 \item $P \stackrel{END}{\rightarrow} \surd$ if and only if $Q \stackrel{END}{\rightarrow} \surd$.
\end{itemize}
\end{defn}

The definition of strong bisimilarity is analogous to definition \ref{def:bisim}, using the new notion of strong bisimulation stated above.

In the presence of the iteration operator, a weaker version of the Expansion Law is now sufficient for our needs.

\begin{thm}[Expansion Law (EL)]\label{teo:XEL}
Let $P = P_1 \mid P_2$, where $P$ is unrestricted \emph{and $|$ does not occur in $P_1$ and $P_2$}. Then
\[
P \sim \sum_{P_1 \stackrel{\alpha}{\rightarrow} P_1'} \alpha.(P_1' \mid P_2) + \sum_{P_2 \stackrel{\beta}{\rightarrow} P_2'} \beta.(P_1 \mid P_2') + \sum_{R \in A_{\tau}} \tau . R  + E_P,
\]
where $A_{\tau} = \{ (P_1' \mid P_2') : P_1 \stackrel{a}{\rightarrow} P_1' \mbox{ and } P_2 \stackrel{\overline{a}}{\rightarrow} P_2', \mbox{ for some } a \in {\cal N} \} \cup \{ (P_1' \mid P_2') : P_1 \stackrel{\overline{a}}{\rightarrow} P_1' \mbox{ and } P_2 \stackrel{a}{\rightarrow} P_2', \mbox{ for some } a \in {\cal N} \}$ and $E_P = END$, if $P_1 \stackrel{END}{\rightarrow} \surd$ and $P_2 \stackrel{END}{\rightarrow} \surd$ or $E_P = {\bf 0}$, otherwise. Again, we denote the right side of this bisimilarity by $Exp(P)$.
\end{thm}

Finally, because of the presence of the action $END$, we need to slightly adjust the definition of the composition $R \circ S$ of two sets $R$ and $S$ of finite sequences of actions.

\begin{defn}
Let $\natural (\overrightarrow{\alpha}) = \overrightarrow{\lambda}$, if $\overrightarrow{\alpha} = \overrightarrow{\lambda} . END$ and $\natural (\overrightarrow{\alpha}) = \overrightarrow{\alpha}$, otherwise. Then,
\[
R \circ S = \{ \natural (\overrightarrow{\alpha}) . \overrightarrow{\beta} : \overrightarrow{\alpha} \in R \,\, \textrm{and} \,\, \overrightarrow{\beta} \in S \}
\]
\end{defn}

Now, we define some concepts that are useful to the axiomatization of our third logic.

\begin{defn}\label{def:cong}
We say that a relation $\cong$ between processes is a \emph{congruence} if it is an equivalence relation and it is preserved by all of XCCS  operators, that is, if $P \cong Q$, then $\alpha.P \cong \alpha.Q$, $P + R \cong Q + R$ and so on.
\end{defn}

\begin{defn}\label{def:alphacon}
A syntactic substitution of a \emph{restricted name} by a \emph{fresh name} (a name that does not occur in the process specification) in a restriction set $L$ and in every occurrence of the name in the scope of the correspondent restriction $\backslash L$ is called an \emph{alpha conversion}. 
\end{defn}

\begin{defn}\label{def:scon}
\emph{Restriction congruence}, or \emph{r-congruence}, denoted by $\equiv_r$, is a relation between processes defined by the following set of axioms and rules, where $n(P)$ denotes the set of names that occur in $P$ as part of both input and output actions.
\begin{multicols}{2}
\begin{enumerate}
\item It is a congruence;
\item It is closed under alpha conversion;
\item ${\bf 0} \backslash L \equiv_r {\bf 0}$;
\item $END \backslash L \equiv_r END$;
\item If $\alpha \not\in L \cup \overline{L}$, $(\alpha.P) \backslash L \equiv_r \alpha.(P \backslash L)$;
\item If $\alpha \in L \cup \overline{L}$, $(\alpha.P) \backslash L \equiv_r {\bf 0}$; 
\item $(P;Q) \backslash L \equiv_r (P \backslash L);(Q \backslash L)$;
\item $(P+Q) \backslash L \equiv_r (P \backslash L) + (Q \backslash L)$;
\item If $n(P) \cap (L \cup \overline{L}) = \emptyset$, $P | (Q \backslash L)$ $\equiv_r (P|Q) \backslash L$;
\item $(P^*) \backslash L \equiv_r (P \backslash L)^*$;
\item $P \backslash L \backslash M \equiv_r P \backslash (L \cup M)$;
\item If $n(P) \cap (L \cup \overline{L}) = \emptyset$, $P \backslash L \equiv_r P$.
\end{enumerate}
\end{multicols}
\end{defn}

\begin{defn}\label{def:rext}
We say that a process is in \emph{r-external form} if it has the form $P \backslash L$, where $P$ is unrestricted.
\end{defn}

\begin{thm}\label{teo:sc}
Every process is r-congruent to a process in r-external form and every process with no occurrences of the $|$ operator is r-congruent to an unrestricted process.
\end{thm}

\begin{proof}
The proof follows from definition \ref{def:scon}.
\end{proof}

\begin{thm}\label{teo:equiv}
If $P \equiv_r Q$, then $P \sim Q$.
\end{thm}

\begin{proof}
The proof follows from table \ref{tab:semxccs} and definition \ref{def:bisim}.
\end{proof}

\subsection{Language and Semantics}

In this section, we present the syntax and semantics of XCCS-PDL.

\begin{defn}\label{def-langxccs} 
The XCCS-PDL language consists of a set $\Phi$ of countably many proposition symbols, a set ${\cal N}$ of countably many names, the silent action $\tau$, the ending action $END$, the boolean connectives $\neg$ and $\land$, the XCCS operators $.$, $;$, $+$, $\mid$, $\phantom{}^*$ and $\backslash$, a modality $\langle \alpha \rangle$ for every $\alpha \in {\cal N} \cup \overline{\cal N} \cup \{ \tau \}$ and a modality $\langle P \rangle$ for every process $P$, including the atomic processes ${\bf 0}$ and $END$. The formulas are defined as follows:
\[
\varphi ::= p \mid \top \mid \neg \varphi \mid \varphi_1 \wedge \varphi_2 \mid \langle \alpha \rangle \varphi \mid \langle P \rangle \varphi,
\]
with
\[ 
P ::= {\bf 0} \mid END \mid \alpha . P \mid P_1 ; P_2 \mid P_1 + P_2 \mid P_1|P_2 \mid P^* \mid P \backslash L,
\]
where $p \in \Phi$, $\alpha \in {\cal N} \cup \overline{\cal N} \cup \{ \tau \}$ and $L \subseteq {\cal N}$.
\end{defn} 

\begin{defn}\label{def:xccsframe} 
A \emph{frame} for XCCS-PDL is a tuple $\mathcal{F}= (W, \{R_{\alpha}\}, R_{END})$ where
\begin{itemize}
\item $W$ is a non-empty set of states;
\item $R_{\alpha}$, for each $\alpha \in {\cal N} \cup \overline{\cal N} \cup \{ \tau \}$ and $R_{END}$ are the basic binary relations, where $R_{END} = \{(w,w) : w \in W\}$.
\end{itemize}
\end{defn}

The notion of model is defined analogously to definition \ref{def-model*}. We define the semantical notion of satisfaction for XCCS-PDL as follows:

\begin{defn}\label{def:xccssat} 
Let $\mathcal{M}= ({\cal F}, {\bf V})$ be a model. The notion of \emph{satisfaction} of a formula $\varphi$ in a model $\mathcal{M}$ at a state $w$, notation $\mathcal{M},w \Vdash \varphi$, can be inductively defined as follows:
\begin{itemize}
\item $\mathcal{M},w \Vdash p$ iff $w \in {\bf V}(p)$;
\item $\mathcal{M},w \Vdash \top$ always;
\item $\mathcal{M},w \Vdash \neg \varphi$ iff $\mathcal{M},w \not\Vdash \varphi$;
\item $\mathcal{M},w \Vdash \varphi_{1} \wedge \varphi_{2}$ iff $\mathcal{M},w \Vdash \varphi_{1}$ and $\mathcal{M},w \Vdash \varphi_{2}$;
\item $\mathcal{M},w \Vdash \langle \alpha \rangle \varphi$ iff there is $w' \in W$ such that $w R_{\alpha} w'$ and $\mathcal{M},w' \Vdash \varphi$, where $\alpha \in {\cal N} \cup \overline{N} \cup \{ \tau \}$;
\item $\mathcal{M},w \Vdash \langle P \rangle \varphi$ iff there is a finite path $(v_0,v_1,\ldots,$ $v_n)$, $n \geq 1$, such that $v_0 = w$, $\mathcal{M},v_n \Vdash \varphi$ and there is $\overrightarrow{\alpha} \in \overrightarrow{{\cal R}_f}(P)$ of length $n$ such that $(v_{i-1},v_{i}) \in R_{\beta}$ if and only if $(\overrightarrow{\alpha})_i = \beta$, for $1 \leq i \leq n$. We say that such $\overrightarrow{\alpha}$ \emph{matches} the path $(v_0,\ldots,v_n)$.
\end{itemize}
\end{defn}

It is not difficult to see that theorem \ref{teo:valPQ} and corollary \ref{cor:simPQvalPQ} remain valid in XCCS-PDL.

\begin{thm}\label{teo:Pfxccs}
The following set equalities are true:
\begin{multicols}{2}
\begin{enumerate}
 \item $\overrightarrow{{\cal R}_f}({\bf 0}) = \emptyset$;
 \item $\overrightarrow{{\cal R}_f}(END) = \{ END \}$;
 \item $\overrightarrow{{\cal R}_f}(\alpha . P) = \overrightarrow{{\cal R}_f}(\alpha) \circ \overrightarrow{{\cal R}_f}(P)$;
 \item $\overrightarrow{{\cal R}_f}(P_1 ; P_2) = \overrightarrow{{\cal R}_f}(P_1) \circ \overrightarrow{{\cal R}_f}(P_2)$;
 \item $\overrightarrow{{\cal R}_f}(P_1 + P_2) = \overrightarrow{{\cal R}_f}(P_1) \cup \overrightarrow{{\cal R}_f}(P_2)$;
 \item $\overrightarrow{{\cal R}_f}(P^*) = \overrightarrow{{\cal R}_f}(P)^*$;
 \item $\overrightarrow{{\cal R}_f}(P_1 | P_2) = \bigcup \{ \overrightarrow{{\cal R}_f}(\overrightarrow{\alpha} | \overrightarrow{\beta}) : \overrightarrow{\alpha} \in \overrightarrow{{\cal R}_f}(P_1) \mbox{ and } \overrightarrow{\beta} \in \overrightarrow{{\cal R}_f}(P_2) \}$;
 \item If $\overrightarrow{{\cal R}_f}(P) = \overrightarrow{{\cal R}_f}(Q)$, then $\overrightarrow{{\cal R}_f}$ $(P \backslash L) = \overrightarrow{{\cal R}_f}(Q \backslash L)$;
 \item If $\overrightarrow{{\cal R}_f}(P) = \overrightarrow{{\cal R}_f}(A) \circ \overrightarrow{{\cal R}_f}(P) \cup \overrightarrow{{\cal R}_f}(B)$ and $END \not\in \overrightarrow{{\cal R}_f}(A)$, then $\overrightarrow{{\cal R}_f}(P) = \overrightarrow{{\cal R}_f}(A)^* \circ \overrightarrow{{\cal R}_f}(B)$.
\end{enumerate}
\end{multicols}
\end{thm}

\begin{proof}
The proof of the first eight items is straightforward from table \ref{tab:semxccs} and theorem \ref{teo:eqRf}. The ninth item is Arden's Rule \cite{Arden} applied in our context.
\end{proof}

\subsection{Axiomatic System}

We consider the following set of axioms and rules, where $p$ and $q$ are proposition symbols and $\varphi$ and $\psi$ are formulas.

\begin{description}
\item[(sCCS)] The axioms {\bf (PL)}, {\bf (K)}, {\bf (Du)}, {\bf (Pr)} and {\bf (NC)} and the rules {\bf (PC)}, {\bf (Sub)}, {\bf (MP)} and {\bf (Gen)}
\item[(0)] $\vdash \neg \langle {\bf 0} \rangle p$
\item[(END)] $\vdash \langle END \rangle p \leftrightarrow p$
\item[(SC)] $\vdash \langle P_1 ; P_2 \rangle p \leftrightarrow \langle P_1 \rangle \langle P_2 \rangle p$
\item[(Rec)] $\vdash \langle P^* \rangle p \leftrightarrow p \lor \langle P \rangle \langle P^* \rangle p$
\item[(FP)] $\vdash p \land [P^*](p \rightarrow [P] p) \rightarrow [P^*] p$
\item[(PCSub)] If $\vdash \langle P \rangle p \leftrightarrow \langle Q \rangle p$, then $\vdash \langle P | R \rangle p \leftrightarrow \langle Q | R \rangle p$
\item[(RSub)] If $\vdash \langle P \rangle p \leftrightarrow \langle Q \rangle p$, then $\vdash \langle P \backslash L \rangle p \leftrightarrow \langle Q \backslash L \rangle p$
\item[(Ard)] If $\vdash \langle P \rangle p \leftrightarrow \langle A;P + B \rangle p$ and $A \stackrel{END}{\not\rightarrow} \surd$, then $\vdash \langle P \rangle p \leftrightarrow \langle A^*;B \rangle p$
\item[(Con)] If $P \equiv_r Q$, then $\vdash \langle P \rangle p \leftrightarrow \langle Q \rangle p$
\end{description}

The proof of soundness is analogous to the proof of soundness for sCCS-PDL and CCS-PDL. The axioms {\bf (PL)}, {\bf (K)} and {\bf (Du)} and the rules {\bf (Sub)}, {\bf (MP)} and {\bf (Gen)} are standard in the modal logic literature. The soundness of {\bf (Pr)}, {\bf (NC)}, {\bf (0)}, {\bf (SC)}, {\bf (Rec)}, {\bf (FP)}, {\bf (PCSub)}, {\bf (RSub)} and {\bf (Ard)} follows from the set equalities in theorem \ref{teo:Pfxccs} and theorem \ref{teo:valPQ}. The soundness of {\bf (END)} also follows from the two previous results with the help of definition \ref{def:xccsframe}. The soundness of {\bf (PC)} and {\bf (Con)} follows from theorems \ref{teo:XEL} and \ref{teo:equiv} with the help of corollary \ref{cor:simPQvalPQ}. The only rule that may require special attention is {\bf (PCSub)}.

\begin{thm}
The rule {\bf (PCSub)} is sound.
\end{thm}

\begin{proof}
By theorem \ref{teo:Pfxccs}, $\overrightarrow{{\cal R}_f}(P_1 | P_2) = \bigcup \{ \overrightarrow{{\cal R}_f}(\overrightarrow{\alpha} | \overrightarrow{\beta}) : \overrightarrow{\alpha} \in \overrightarrow{{\cal R}_f}(P_1) \,\, \textrm{and} \,\, \overrightarrow{\beta} \in \overrightarrow{{\cal R}_f}(P_2) \}$. Now, suppose that $\Vdash \langle P \rangle p \leftrightarrow \langle Q \rangle p$, but $\not\Vdash \langle P | R \rangle p \leftrightarrow \langle Q | R \rangle p$. Then, by theorem \ref{teo:valPQ}, $\overrightarrow{{\cal R}_f}(P) = \overrightarrow{{\cal R}_f}(Q)$, but $\overrightarrow{{\cal R}_f}(P | R) \neq \overrightarrow{{\cal R}_f}(Q | R)$. We may assume, without loss of generality, that there is $\overrightarrow{\lambda}$ such that $\overrightarrow{\lambda} \in \overrightarrow{{\cal R}_f}(P | R)$ (*), but $\overrightarrow{\lambda} \not\in \overrightarrow{{\cal R}_f}(Q | R)$ (**). (*) implies that there is $\overrightarrow{\alpha} \in \overrightarrow{{\cal R}_f}(P)$ and $\overrightarrow{\beta} \in \overrightarrow{{\cal R}_f}(R)$ such that $\overrightarrow{\lambda} \in \overrightarrow{{\cal R}_f}(\overrightarrow{\alpha} | \overrightarrow{\beta})$. But then $\overrightarrow{\alpha} \in \overrightarrow{{\cal R}_f}(Q)$, which implies that $\overrightarrow{\lambda} \in \overrightarrow{{\cal R}_f}(Q | R)$, contradicting (**).
\end{proof}

\begin{defn}
We define the following relation between processes: $P \leftrightarrow Q$ iff $\vdash \langle P \rangle p \leftrightarrow \langle Q \rangle p$.
\end{defn}

\begin{thm}\label{teo:axcong}
$\leftrightarrow$ is a \emph{congruence}.
\end{thm} 

\begin{proof}
This relation is clearly an equivalence relation and the axioms {\bf (Pr)}, {\bf (SC)}, {\bf (NC)}, {\bf (Rec)} and {\bf (FP)} and the rules {\bf (PCSub)} and {\bf (RSub)} enforce the preservation results needed to satisfy definition \ref{def:cong}.
\end{proof}

\begin{defn}\label{def:eqpar}
Let $\Omega_k = \{P_1,\ldots,P_k\}$ be a set of processes such that $P_i \not\equiv_r P_j$, if $i \neq j$. Let $E(\Omega_k) = \{E_1,\ldots,E_k\}$ such that $ E_i = (P_i,T_i)$, $P_i \leftrightarrow T_i$, $T_i = \sum_j A^i_j ; Q^i_j$ and, for all $(i,j)$, $A^i_j$ has no occurrence of $|$. We say that $E(\Omega_k)$ is \emph{closed} if, for all $(i,j)$, $Q^i_j \in \Omega_k$.   
\end{defn}

\begin{thm}\label{teo:tirapar}
Let $P = P_1 \mid P_2$, where $P$ is unrestricted. Then $P \leftrightarrow \overline{P}$, where $\overline{P}$ has no occurrence of the $|$ operator.
\end{thm}

\begin{proof}
The proof is by induction on the number $n$ of occurrences of the $|$ operator in $P$. If $n = 0$, then $\overline{P} = P$ and there is nothing to be done. 

If $n = 1$, then EL can be applied to $P$. Then, we can use {\bf (PC)} to build pairs $(P_i,T_i)$ that satisfy definition \ref{def:eqpar}. Let $P_1 = P$ and $\Omega_k$ be the smallest set such that $P_1 \in \Omega_k$ and $E(\Omega_k)$ is closed. It is not difficult to see that such set always exist. Take the pair $E_k$. If there is no $Q^k_j = P_k$ (*), then we can substitute in the processes $T_i$, $1 \leq i < k$, all the occurrences of $P_k$ by $T_k$. Otherwise, we can use {\bf (Ard)} to substitute the pair $(P_k,T_k)$ by a pair $(P_k,T_k')$ where (*) holds and then proceed as in the previous case. We then continue this process with the pair $E_{k-1}$ and so on, until we finally get a pair $(P_1,T_1')$ such that no process in $\Omega_k$ occurs in $T_1'$. By the use of {\bf (PC)} to build the initial pairs and the fact that neither {\bf (Ard)} nor the substitution process introduce new $|$ operators, we have $\overline{P} = T_1'$. This method, based on the solution of a ``system of equations'', was inspired by Brzozowski's algebraic method to obtain the regular expression that describes the language accepted by a finite automaton \cite{Brzo}.

Suppose that the theorem is true for all $n < k$. Let $P$ have $k$ occurrences of $|$. As $P = P_1 | P_2$, we can obtain $\overline{P}$ as $\overline{\overline{P_1} | \overline{P_2}}$.
\end{proof}

Two formulas $\phi$ and $\psi$ are equi-consistent if $\vdash \phi \leftrightarrow \psi$. By soundness, if $\phi$ and $\psi$ are equi-consistent, then they are also semantically equivalent.

\begin{thm}[Completeness]
Every consistent formula is satisfiable in a finite XCCS-PDL model.
\end{thm}

\begin{proof}
Let $\varphi$ be a consistent formula and let ${\bf P}(\varphi)$ be the set of processes that appear in $\varphi$. For all $P \in {\bf P}(\varphi)$, we can use {\bf (Con)}, {\bf (RSub)} and theorems \ref{teo:sc}, \ref{teo:axcong} and \ref{teo:tirapar} to get a sequence $P \leftrightarrow P' \leftrightarrow P'' \leftrightarrow P'''$, where $P'$ is r-external form, $P''$ is also without any occurrence of the $|$ operator and $P'''$ is like $P''$ but unrestricted. We can then obtain an equi-consistent formula $\varphi' = \varphi[P'''/P, P \in {\bf P}(\varphi)]$ in which the only XCCS operators that appear are $.$, $;$, $+$ and $\phantom{}^*$. The axioms that deal with all of these operators are analogous to the axioms that deal with the operators in standard PDL. {\bf (Pr)} and {\bf (SC)} are analogous to the axiom of the PDL $;$ operator, {\bf (NC)} is analogous to the axiom of the PDL $\cup$ operator and {\bf (Rec)} and {\bf (FP)} are analogous to the axioms of the PDL $\phantom{}^*$ operator. Thus, we can follow the completeness proof of standard PDL (the PDL axioms and its completeness proof are presented in details in \cite{Yde}), treating the actions as basic PDL programs, to show that $\varphi'$ is satisfiable in a finite model. As $\varphi$ and $\varphi'$ are equi-consistent, they are also semantically equivalent, which means that $\varphi$ is also satisfied in that same finite model.
\end{proof}

\section{Final Remarks and Future Work}\label{DC}

In this work, we present three increasingly expressive Dynamic Logics in which the programs are CCS terms (sCCS-PDL, CCS-PDL and XCCS-PDL). We provide a simple Kripke semantics for them, based on the idea of finite possible runs of processes, and also give complete axiomatizations for these logics. We prove the completeness of the axiomatic systems and the finite model property for the logics using a Fischer-Ladner construction.

We also provide a method, in a language with a iteration ($\phantom{}^*$) and sequential composition ($;$) operators, to rewrite any process specification to a form without the parallel composition operator ($|$) while preserving the set of finite possible runs of the process. This method is based on Brzozowski's algorithm to find the regular expression that corresponds to a finite automaton \cite{Brzo}. We feel that this is an interesting and original application of Brzozowski's idea and that it provides an elegant proof to a key result to the completeness of our last axiomatization.

As a continuation of this work, it would be interesting to study the complexity of the satisfiability problem for these logics, possibly relating it to the satisfiability problem for standard PDL. It would also be interesting to develop an automatic theorem prover for these logics. This would involve, among other things, an efficient algorithmic method to deal with the expansion of parallel processes and, in the particular case of CCS-PDL, an efficient algorithmic method to determine the processes $L_P$ and $T_P$ related to a knot process $P$.

We would also like to investigate an extension of these logics for $\pi$-Calculus processes \cite{M99}, in which the acts of communications are more complex than in CCS. The $\pi$-Calculus is a very powerful process algebra that is able to describe not only non-determinism and concurrency, but also \emph{mobility} of processes and that can also be used to encode some powerful programming paradigms, as object-oriented programming and functional programming ($\lambda$-Calculus) \cite{M99}. Besides that, the $\pi$-Calculus has a specific operator to denote that a process has the ability to self-replicate, so this could be an interesting context to analyze in more depth the issue of self-replicating processes, which was left out of the present work.

\appendix

\section{Completeness Proof for CCS-PDL}\label{sec:compproof}

\begin{defn}\label{def:over}
Let $\phi$ be a formula. We define the formula $\overline{\phi}$ as $\overline{\phi} = \psi$, if $\phi = \neg \psi$, or $\overline{\phi} = \neg \phi$, otherwise.
\end{defn}

\begin{defn}[Fischer-Ladner Closure]\label{def:hfl}
Let $\Gamma$ be a set of formulas. The \emph{Fischer-Ladner Closure} of $\Gamma$, notation $C(\Gamma)$, is the smallest set of formulas that contains $\Gamma$ and satisfies the following conditions:

\begin{itemize}
\item $C(\Gamma)$ is closed under sub-formulas;
\item if $\phi \in C(\Gamma)$, then $\overline{\phi} \in C(\Gamma)$;
\item For knot processes:
\begin{itemize}
\item If $\langle P \rangle \varphi \in C(\Gamma)$, then $\langle T_P \rangle \varphi \lor \langle L_P \rangle \langle P \rangle \varphi \in C(\Gamma)$.
\end{itemize}
\item For non-knot processes:
\begin{itemize}
\item If $\langle\alpha . P \rangle\varphi \in C(\Gamma)$, then $\langle \alpha \rangle\langle P \rangle\varphi \in C(\Gamma)$;
\item If $\langle\alpha . A \rangle\varphi \in C(\Gamma)$, then $\langle \alpha \rangle\langle P_A \rangle\varphi \in C(\Gamma)$;
\item If $\langle P_{1} + P_{2} \rangle\varphi \in C(\Gamma)$, then $\langle P_{1} \rangle\varphi \lor \langle P_{2} \rangle\varphi  \in C(\Gamma)$;
\item If $\langle P_{1} \mid P_{2} \rangle \varphi \in C(\Gamma)$, then $\bigvee_{P_{1} \stackrel{\alpha}{\rightarrow} P_{1}'} \langle \alpha \rangle \langle P_{1}' \mid P_{2} \rangle \varphi \lor \bigvee_{P_{2} \stackrel{\alpha}{\rightarrow} P_{2}'} \langle \alpha \rangle \langle P_{1} \mid P_{2}' \rangle \varphi \lor \bigvee_{\stackrel{P_{1} \stackrel{\lambda}{\rightarrow} P_{1}'}{P_{2} \stackrel{\overline{\lambda}}{\rightarrow} P_{2}'}} \langle \tau \rangle \langle P_{1}' \mid P_{2}' \rangle \varphi \in C(\Gamma)$.
\end{itemize}
\end{itemize}
\end{defn}

It is not difficult to prove that if $\Gamma$ is finite, then the closure $C(\Gamma)$ is also finite. We assume $\Gamma$ to be finite from now on.

\begin{defn}\label{def:atom}
A set of formulas ${\cal A}$ is said to be an \emph{atom} over $\Gamma$ if it is a maximal consistent subset of $C(\Gamma)$. The set of all atoms over $\Gamma$ is denoted by $At(\Gamma)$. We denote the conjunction of all the formulas in an atom ${\cal A}$ as $\bigwedge {\cal A}$. 
\end{defn}

\begin{lem}\label{lem:atomprop}
Every atom ${\cal A} \in At(\Gamma)$ has the following properties:
\begin{enumerate}
\item For every $\phi \in C(\Gamma)$, exactly one of $\phi$ and $\neg \phi$ belongs to ${\cal A}$.
\item For every $\phi \land \psi \in C(\Gamma)$, $\phi \land \psi \in {\cal A}$ iff $\phi \in {\cal A}$ and $\psi \in {\cal A}$.
\end{enumerate}
\end{lem}

\begin{proof}
This follows immediately from the definition of atoms as maximal consistent subsets of $C(\Gamma)$.
\end{proof}

\begin{lem}\label{lem:expatom}
If $\Delta \subseteq C(\Gamma)$ and $\Delta$ is consistent
then there exists an atom ${\cal A} \in At(\Gamma)$ such that $\Delta \subseteq {\cal A}$.
\end{lem}

\begin{proof}
We can construct the atom ${\cal A}$ as follows. First, we enumerate the elements
of $C(\Gamma)$ as $\phi_{1}, \ldots , \phi_{n}$. We start the construction making
${\cal A}_{0} = \Delta$. Then, for $0 \leq i < n$, we know that
$\bigwedge {\cal A}_{i} \leftrightarrow  (\bigwedge {\cal A}_{i} \land \phi_{i+1})
\lor (\bigwedge {\cal A}_{i} \land \overline{\phi_{i+1}})$ is a tautology and therefore
either ${\cal A}_{i} \cup \{ \phi_{i+1} \}$ or ${\cal A}_{i} \cup \{ \overline{\phi_{i+1}} \}$ is
consistent. We take ${\cal A}_{i+1}$ as the consistent extension. At the end, we make ${\cal A} = {\cal A}_{n}$.
\end{proof}

\begin{cor}\label{cor:expalpha}
If $\varphi \in C(\Gamma)$ is a consistent formula, then there is an atom ${\cal A} \in At(\Gamma)$ such that $\varphi \in {\cal A}$.
\end{cor}

\begin{defn}[Canonical model over $\Gamma$]\label{def:canpremod}
Let $\Gamma$ be a finite set of formulas. The \emph{canonical model} over $\Gamma$ is the tuple ${\cal M}^{\Gamma} = (At(\Gamma), \{S_{\alpha}\}, {\bf V})$ where, for all elements $p \in \Phi$, we have ${\bf V}(p) = \{ {\cal A} \in At(\Gamma) \mid p \in {\cal A}\}$ and for all atoms ${\cal A}, {\cal B} \in At(\Gamma)$,
\[
{\cal A}S_{\alpha}{\cal B} \,\, \textrm{iff} \,\, \bigwedge {\cal A} \land \langle \alpha \rangle \bigwedge {\cal B} \,\, \textrm{is consistent}.
\]
${\bf V}$ is called the \emph{canonical valuation} and $S_{\alpha}$ the \emph{canonical relations}, where $\alpha$ is a CCS action.
\end{defn}

\begin{defn}
We write ${\cal A} \stackrel{P}{\twoheadrightarrow} {\cal B}$ if and only if $\bigwedge {\cal A} \land \langle P \rangle \bigwedge {\cal B}$ is consistent. We also write $S_P = \{ ({\cal A},{\cal B}) : {\cal A} \stackrel{P}{\twoheadrightarrow} {\cal B} \}$.
\end{defn}

\begin{lem}[Existence Lemma for Basic Processes]\label{lem:existence}
Let ${\cal A}$ be an atom and let $\alpha$ be an action. Then, for all formulas $\langle \alpha \rangle \phi \in C(\Gamma)$, $\langle \alpha \rangle \phi \in {\cal A}$ iff there is a ${\cal B} \in At(\Gamma)$ such that ${\cal A} S_{\alpha} {\cal B}$ and $\phi \in {\cal B}$.
\end{lem}

\begin{proof}
($\Rightarrow$) Suppose $\langle \alpha \rangle \phi \in {\cal A}$. We can build an appropriate atom ${\cal B}$ by forcing choices. Enumerate the formulas in $C(\Gamma)$ as $\phi_1, \ldots, \phi_n$. Define ${\cal B}_0 = \{\phi\}$. Suppose, as an inductive hypothesis that ${\cal B}_m$ is defined such that $\bigwedge {\cal A} \land \langle \alpha \rangle \bigwedge {\cal B}_m$ is consistent, for $0 \leq m < n$. We have that 
\[
\vdash \langle \alpha \rangle \bigwedge {\cal B}_{m} \leftrightarrow \langle \alpha \rangle((\bigwedge {\cal B}_{m} \land \phi_{m+1}) \lor (\bigwedge {\cal B}_{m} \land \overline{\phi_{m+1}})),
\]
thus
\[
\vdash \langle \alpha \rangle \bigwedge {\cal B}_{m} \leftrightarrow  (\langle \alpha \rangle(\bigwedge {\cal B}_{m} \land \phi_{m+1}) \lor \langle \alpha \rangle(\bigwedge {\cal B}_{m} \land \overline{\phi_{m+1}})).
\]
Therefore, either for ${\cal B}' = {\cal B}_m \cup \{ \phi_{m+1} \}$ or for ${\cal B}' = {\cal B}_m \cup \{ \overline{\phi_{m+1}} \}$, we have that $\bigwedge {\cal A} \land \langle \alpha \rangle \bigwedge {\cal B}'$ is consistent. We take ${\cal B}_{m+1}$ as the consistent extension. At the end, we make ${\cal B} = {\cal B}_{n}$. We have that $\phi \in {\cal B}$ and, as $\bigwedge {\cal A} \land \langle \alpha \rangle \bigwedge {\cal B}$ is consistent, ${\cal A} S_{\alpha} {\cal B}$, by definition \ref{def:canpremod}.

($\Leftarrow$): Suppose that there is an atom ${\cal B}$ such that $\phi \in {\cal B}$ and ${\cal A} S_{\alpha} {\cal B}$. Then $\bigwedge {\cal A} \land \langle \alpha \rangle \bigwedge {\cal B}$ is consistent by definition \ref{def:canpremod}. As $\phi$ is one of the conjuncts of $\bigwedge {\cal B}$, $\bigwedge {\cal A} \land \langle \alpha \rangle \phi$ is also consistent. As $\langle \alpha \rangle \phi$ is in $C(\Gamma)$, it must also be in ${\cal A}$, since ${\cal A}$ is a maximal consistent subset of $C(\Gamma)$.
\end{proof}

\begin{lem}\label{lem:loop}
For all knot processes $P$, $S_P \subseteq S_P'$, where $S_P' = S_{L_P}^{*} \circ S_{T_P}$. 
\end{lem}

\begin{proof}
For an atom ${\cal B} \in At(\Gamma)$ and a relation $S$, we denote the set of atoms $\{{\cal A} \mid {\cal A} S {\cal B}\}$ as $\langle S \rangle {\cal B}$. Suppose there are two atoms ${\cal A}, {\cal B} \in At(\Gamma)$ such that ${\cal A} \in \langle S_P \rangle {\cal B}$, but ${\cal A} \notin \langle S_P' \rangle {\cal B}$. Let $V = \{{\cal C} \in At(\Gamma) \mid {\cal C} \in \langle S_P \rangle {\cal B} \, \mbox{ but } \, {\cal C} \notin \langle S_P' \rangle {\cal B}\} \cup \{{\cal C} \in At(\Gamma) \mid {\cal C} \notin \langle S_P \rangle {\cal B}\}$ and $\overline{V} = At(\Gamma) \setminus V = \{{\cal C} \in At(\Gamma) \mid {\cal C} \in \langle S_P \rangle {\cal B} \, \mbox{ and } \, {\cal C} \in \langle S_P' \rangle {\cal B}\}$. Thus, ${\cal A} \in V$. Let $r = \bigvee \{ \bigwedge {\cal C} \mid {\cal C} \in V\}$. It is not difficult to see that $\neg r = \bigvee \{ \bigwedge {\cal C} \mid {\cal C} \in \overline{V}\}$. 

First, we have that $\vdash r \rightarrow [T_P] \neg \bigwedge {\cal B}$. Otherwise, $\neg (r \rightarrow [T_P] \neg \bigwedge {\cal B}) \equiv r \land \langle T_P \rangle \bigwedge {\cal B}$ is consistent. This means that there is ${\cal A}' \in V$ such that $\bigwedge {\cal A}' \land \langle T_P \rangle \bigwedge {\cal B}$ is consistent. On one hand, this implies, by {\bf (Rec)}, that $\bigwedge {\cal A}' \land \langle P \rangle \bigwedge {\cal B}$ is consistent, which means that ${\cal A}' \in \langle S_P \rangle {\cal B}$. On the other hand, it implies that ${\cal A}' S_{T_P} {\cal B}$, which means that ${\cal A}' \in \langle S_P' \rangle {\cal B}$. These two conclusions contradict the fact that ${\cal A}' \in V$.

Second, we also have that $\vdash r \rightarrow [L_P]r$. Otherwise, $\neg (r \rightarrow [L_P]r) \equiv r \land \langle L_P \rangle \neg r$ is consistent. This means that there are ${\cal A}' \in V$ and ${\cal B'} \in \overline{V}$ such that $\bigwedge {\cal A}' \land \langle L_P \rangle \bigwedge {\cal B}'$ is consistent, which implies that ${\cal A}' S_{L_P} {\cal B}'$. Since ${\cal B}' \in \overline{V}$, ${\cal B}' S_P {\cal B}$ and ${\cal B}' S_P' {\cal B}$. On one hand, ${\cal A}' S_{L_P} {\cal B}'$ and ${\cal B}' S_P' {\cal B}$ imply that ${\cal A}' S_P' {\cal B}$ (*). On the other hand, ${\cal A}' S_{L_P} {\cal B}'$ and ${\cal B}' S_P {\cal B}$ imply that $\bigwedge {\cal A}' \land \langle L_P \rangle \langle P \rangle \bigwedge {\cal B}$ is consistent, which, by {\bf (Rec)}, implies that $\bigwedge {\cal A}' \land \langle P \rangle \bigwedge {\cal B}$ is consistent, which means that ${\cal A}' S_P {\cal B}$ (**). The conclusions in (*) and (**) contradict the fact that ${\cal A}' \in V$.

Taking these two results together, we conclude that $\vdash r \rightarrow ([T_P]\neg \bigwedge {\cal B} \land [L_P]r)$. By {\bf (Gen)}, {\bf (PL)}, {\bf (FP)} and {\bf (MP)}, $\vdash r \rightarrow [P] \neg \bigwedge {\cal B}$. But, as ${\cal A} \in V$, $\vdash \bigwedge {\cal A} \rightarrow r$, which means that $\vdash \bigwedge {\cal A} \rightarrow [P] \neg \bigwedge {\cal B}$. This implies that $\bigwedge {\cal A} \land \langle P \rangle \bigwedge {\cal B}$ is inconsistent, contradicting the fact that ${\cal A} S_P {\cal B}$. Thus, there cannot be a pair of atoms ${\cal A}, {\cal B} \in At(\Gamma)$ such that ${\cal A} \in \langle S_P \rangle {\cal B}$, but ${\cal A} \notin \langle S_P' \rangle {\cal B}$.
\end{proof}

\begin{defn}
We write ${\cal A} \stackrel{P}{\rightsquigarrow} {\cal B}$ if and only if there is a path in the canonical model starting in ${\cal A}$ and ending in ${\cal B}$ such that there is $\overrightarrow{\alpha} \in \overrightarrow{{\cal R}_f}(P)$ that matches it. We also write $R_P = \{ ({\cal A},{\cal B}) : {\cal A} \stackrel{P}{\rightsquigarrow} {\cal B} \}$. Finally, it also follows from this definition that ${\cal M}^{\Gamma}, {\cal A} \Vdash \langle P \rangle \varphi$ if and only if there is ${\cal B}$ such that $({\cal A},{\cal B}) \in R_P$ and ${\cal M}^{\Gamma}, {\cal B} \Vdash \varphi$.
\end{defn}

\begin{lem}\label{lemma:S-implis-R} 
For all processes $P$, $S_P \subseteq R_P$.
\end{lem}

\begin{proof}
The proof is by induction on the structure of the process $P$.
\begin{itemize}
\item If $P$ is an action $\alpha$, then the proof is straightforward. First, $\overrightarrow{{\cal R}_f}(P) = \{ \alpha \}$. Now, if ${\cal A} S_{\alpha} {\cal B}$, then there is a path in the canonical model starting in ${\cal A}$ and ending in ${\cal B}$ such that there is $\overrightarrow{\alpha} \in \overrightarrow{{\cal R}_f}(P)$ that matches it. Hence, ${\cal A} R_{\alpha} {\cal B}$ is true as well.

\item $P$ is a non-knot process:
\begin{itemize}
\item Suppose ${\cal A} S_{\alpha . P} {\cal B}$, that is, $\bigwedge {\cal A} \wedge \langle \alpha . P \rangle \bigwedge {\cal B}$ is consistent. By {\bf (Pr)}, $\bigwedge {\cal A} \wedge \langle \alpha \rangle\langle P \rangle \bigwedge {\cal B}$ is consistent as well. Using  a ``forcing choices'' argument (as exemplified in lemma \ref{lem:existence}), we can construct an atom ${\cal C}$ such that $\bigwedge {\cal A} \wedge \langle \alpha \rangle \bigwedge {\cal C}$ and $\bigwedge {\cal C} \wedge \langle P \rangle \bigwedge {\cal B}$ are both consistent. But then, by the inductive hypothesis, ${\cal A} R_{\alpha} {\cal C}$ and ${\cal C} R_{P} {\cal B}$. It follows that  ${\cal A} R_{\alpha . P} {\cal B}$ as required. 

\item Suppose ${\cal A} S_{\alpha . A} {\cal B}$, that is, $\bigwedge {\cal A} \wedge \langle \alpha . A \rangle \bigwedge {\cal B}$ is consistent. By {\bf (Cons)}, $\bigwedge {\cal A} \wedge \langle \alpha \rangle\langle P_A \rangle \bigwedge {\cal B}$ is consistent as well. Using  a ``forcing choices'' argument, we can construct an atom ${\cal C}$ such that $\bigwedge {\cal A} \wedge \langle \alpha \rangle \bigwedge {\cal C}$ and $\bigwedge {\cal C} \wedge \langle P_A \rangle \bigwedge {\cal B}$ are both consistent. But then, by the inductive hypothesis, ${\cal A} R_{\alpha} {\cal C}$ and ${\cal C} R_{P_A} {\cal B}$. It follows that  ${\cal A} R_{\alpha . A} {\cal B}$ as required. 

\item Suppose ${\cal A} S_{P_{1} + P_{2}} {\cal B}$, that is, $\bigwedge {\cal A} \wedge \langle P_{1} + P_{2} \rangle \bigwedge {\cal B}$ is consistent. By {\bf (NC)}, $\bigwedge {\cal A} \wedge \langle P_{1} \rangle \bigwedge {\cal B}$ is consistent or $\bigwedge {\cal A} \wedge \langle P_{2} \rangle \bigwedge {\cal B}$ is consistent. But then, by the inductive hypothesis, ${\cal A} R_{P_{1}} {\cal B}$ or  ${\cal A} R_{P_{2}} {\cal B}$. It follows that ${\cal A} R_{P_{1} + P_{2}} {\cal B}$ as required.

\item Suppose ${\cal A} S_{P_{1} \mid P_{2}} {\cal B}$, that is, $\bigwedge {\cal A} \wedge \langle P_{1} \mid P_{2} \rangle \bigwedge {\cal B}$ is consistent. By {\bf (PC)}, $\bigwedge {\cal A} \wedge \langle \alpha \rangle \langle P' \rangle \bigwedge {\cal B}$ is consistent for some basic process $\alpha$ and some process $P'$. Using  a ``forcing choices'' argument, we can construct an atom ${\cal C}$ such that $\bigwedge {\cal A} \wedge \langle \alpha \rangle \bigwedge {\cal C}$ and $\bigwedge {\cal C} \wedge \langle P' \rangle \bigwedge {\cal B}$ are both consistent. But then, by the inductive hypothesis, ${\cal A} R_{\alpha} {\cal C}$ and ${\cal C} R_{P'} {\cal B}$. It follows that  ${\cal A} R_{\alpha . P'} {\cal B}$, which means that ${\cal A} R_{P_1 \mid P_2} {\cal B}$ as required. 
\end{itemize}

\item Suppose ${\cal A} S_P {\cal B}$, where $P$ is a knot process. By lemma \ref{lem:loop}, $S_P \subseteq S_P'$, where $S_P' = S_{L_P}^{*} \circ S_{T_P}$. By the induction hypothesis, $S_{L_P} \subseteq R_{L_P}$ and $S_{T_P} \subseteq R_{T_P}$. This implies that $S_P' \subseteq R_P$, which proves the result.
\end{itemize}
\end{proof}

\begin{lem}[Existence Lemma]\label{lem:existence2} For all atoms ${\cal A} \in At(\Gamma)$ and all formulas $\langle P \rangle\phi \in C(\Gamma)$, $\langle P \rangle\phi \in {\cal A}$ iff there is ${\cal B} \in At(\Gamma)$ such that ${\cal A} R_{P} {\cal B}$ and $\phi \in {\cal B}$.
\end{lem}

\begin{proof}
($\Rightarrow$) Suppose $\langle P \rangle \phi \in {\cal A}$. We can build an atom ${\cal B}$ such that $\phi \in {\cal B}$ and ${\cal A} S_P {\cal B}$ by ``forcing choices''. But, by lemma \ref{lemma:S-implis-R}, $S_P \subseteq R_P$, thus ${\cal A} R_P {\cal B}$ as well.

($\Leftarrow$) We proceed by induction on the structure of $P$.
\begin{itemize}
\item The base case is just the Existence Lemma for basic processes.

\item $P$ is a non-knot process:
\begin{itemize}
\item Suppose $P$ has the form $\alpha . P'$, ${\cal A} R_{\alpha . P'} {\cal B}$ and $\phi \in {\cal B}$. Thus, there is an atom ${\cal C}$ such that  ${\cal A} R_{\alpha} {\cal C}$ and ${\cal C} R_{P'} {\cal B}$. By the Fischer-Ladner closure conditions, $\langle P' \rangle \phi \in C(\Gamma)$, hence by the induction hypothesis, $\langle P \rangle \phi \in C$. Similarly, as $\langle \alpha \rangle \langle P' \rangle \phi \in C(\Gamma)$, $\langle \alpha \rangle \langle P' \rangle \phi \in {\cal A}$. Hence, by {\bf (Pr)}, $\langle \alpha . P \rangle \phi \in {\cal A}$.

\item Suppose $P$ has the form $\alpha . A$, ${\cal A} R_{\alpha . A} {\cal B}$ and $\phi \in {\cal B}$. Thus, there is an atom ${\cal C}$ such that  ${\cal A} R_{\alpha} {\cal C}$, ${\cal C} R_{P_A} {\cal B}$ and $\phi \in {\cal B}$. By the Fischer-Ladner closure conditions, $\langle P_A \rangle \phi \in C(\Gamma)$, hence by the induction hypothesis, $\langle P_A \rangle \phi \in C$. Similarly, as $\langle \alpha \rangle \langle P_A \rangle \phi \in C(\Gamma)$, $\langle \alpha \rangle \langle P_A \rangle \phi \in {\cal A}$. Hence, by {\bf (Cons)}, $\langle \alpha . A \rangle \phi \in {\cal A}$.

\item Suppose $P$ has the form $P_1 + P_2$, ${\cal A} R_{P_{1} + P_{2}} {\cal B}$ and $\phi \in {\cal B}$. Thus, ${\cal A} R_{P_1} {\cal B}$ or ${\cal A} R_{P_2} {\cal B}$. By the Fischer-Ladner closure conditions, $\langle P_1 \rangle \phi, \langle P_2 \rangle \phi \in C(\Gamma)$, hence by the inductive hypothesis, $\langle P_1 \rangle \phi \in {\cal A}$ or $\langle P_2 \rangle \phi \in {\cal A}$. Hence, by {\bf (NC)}, $\langle P_1 + P_2 \rangle \phi \in {\cal A}$.

\item Suppose $P$ has the form $P_1 \mid P_2$, ${\cal A} R_{P_{1} \mid P_{2}} {\cal B}$ and $\phi \in {\cal B}$. Thus, ${\cal A} R_{\alpha . P'} {\cal B}$ for some process $\alpha$ and some process $P'$. Then, there is an atom ${\cal C}$ such that  ${\cal A} R_{\alpha} {\cal C}$ and ${\cal C} R_{P'} {\cal B}$. By the Fischer-Ladner closure conditions, $\langle \alpha. P' \rangle \phi, \langle \alpha \rangle$ $\langle P' \rangle \phi, \langle P' \rangle \phi \in C(\Gamma)$, hence by the inductive hypothesis, $\langle P \rangle \phi \in C$ and $\langle \alpha \rangle \langle P' \rangle \phi \in {\cal A}$. Hence, by {\bf (Pr)}, $\langle \alpha . P \rangle \phi \in {\cal A}$ and, by {\bf (PC)}, $\langle P_1 \mid P_2 \rangle \phi \in {\cal A}$.
\end{itemize}

\item Suppose $P$ is a knot process, ${\cal A} R_P {\cal B}$ and $\phi \in {\cal B}$. Then, there is a finite sequence  of atoms ${\cal C}_{0} \ldots {\cal C}_{n}$ such that ${\cal A} = {\cal C}_{0} R_{L_P} {\cal C}_{1} \ldots {\cal C}_{n-1} R_{L_P} {\cal C}_{n} R_{T_P} {\cal B}$. We prove by a sub-induction on $n$  that $\langle P \rangle \phi \in {\cal C}_{i}$, for all $i$. The desired result for ${\cal A} = {\cal C}_0$ follows immediately.
\begin{itemize}
\item Base case: $n = 0$. This means ${\cal A} R_{T_P} {\cal B}$. By the Fischer-Ladner closure conditions, $\langle T_P \rangle \phi \in C(\Gamma)$, hence by the inductive hypothesis, $\langle T_P \rangle \phi \in {\cal A}$. Hence, by {\bf (Rec)}, $\langle P \rangle \phi \in {\cal A}$.

\item Inductive step: Suppose the result holds for $k < n$, and that ${\cal A} = {\cal C}_{0} R_{L_P} {\cal C}_{1}$ $\ldots R_{L_P} {\cal C}_{n} R_{T_P} {\cal B}$. By the inductive hypothesis, $\langle P \rangle \phi \in {\cal C}_{1}$. Hence $\langle L_P \rangle \langle P \rangle \phi$ $\in {\cal A}$, as $\langle L_P \rangle \langle P \rangle \phi \in C(\Gamma)$. By {\bf (Rec)}, we have that$\langle P \rangle \phi \in {\cal A}$.
\end{itemize}
\end{itemize}
\end{proof}

\begin{lem}[Truth Lemma]\label{lemma:truth} Let $\mathcal{M}^{\Gamma}=(At(\Gamma), \{S_{\alpha}\}, {\bf V})$ be the canonical model over $\Gamma$. For all atoms ${\cal A} \in At(\Gamma)$ and all formulas $\varphi \in C(\Gamma)$, $\mathcal{M}^{\Gamma},{\cal A} \Vdash \varphi$ iff $\varphi \in {\cal A}$.
\end{lem}

\begin{proof}
The proof is by induction on the structure of the formula $\varphi$.
\begin{itemize}
\item $\phi$ is a proposition symbol: The proof follows directly from the definition of ${\bf V}$.
\item $\phi = \neg \psi$ or $\phi = \psi_1 \land \psi_2$: The proof follows directly from lemma \ref{lem:atomprop}.
\item $\phi = \langle P \rangle \psi$:

($\Rightarrow$) Suppose that ${\cal M}^{\Gamma},{\cal A} \Vdash \langle P \rangle \psi$. Then, there exists ${\cal A}' \in {\cal M}^{\Gamma}$ such that ${\cal A} R_P {\cal A}'$ and ${\cal M}^{\Gamma},{\cal A}' \Vdash \psi$. By the induction hypothesis, we know that $\psi \in {\cal A}'$ and, by the Existence Lemma, we have that $\langle P \rangle \psi \in {\cal A}$.

($\Leftarrow$) Suppose that $\langle P \rangle \psi \in {\cal A}$. Then, by the Existence Lemma, there is ${\cal A}' \in {\cal M}^{\Gamma}$ such that ${\cal A} R_P {\cal A}'$ and $\psi \in {\cal A}'$. By the induction hypothesis, ${\cal M}^{\Gamma},{\cal A}' \Vdash \psi$, which implies ${\cal M}^{\Gamma},{\cal A} \Vdash \langle P \rangle \psi$.
\end{itemize}
\end{proof}

\begin{thm}[Completeness]\label{teo:complete}
Every consistent formula is satisfiable in a finite CCS-PDL model.
\end{thm}

\begin{proof}
Let $\varphi$ be a consistent formula. Let $C(\varphi)$ be its closure under the conditions of definition \ref{def:hfl}. As $\varphi$ is consistent, by corollary \ref{cor:expalpha}, there is an atom ${\cal A} \in At(\varphi)$ such that $\varphi \in {\cal A}$. Let ${\cal M}^{\varphi}$ be the canonical model over $\varphi$. Then, by the Truth Lemma (lemma \ref{lemma:truth}), as $\varphi \in {\cal A}$, we conclude that ${\cal M}^{\varphi}, {\cal A} \Vdash \varphi$, which proves the theorem.
\end{proof}

\end{document}